\documentclass[journal,twocolumn]{IEEEtran}

\usepackage[T1]{fontenc}
\usepackage{microtype}

\usepackage{amsmath}
\usepackage{amssymb}
\usepackage{amsfonts}
\usepackage{mathtools}
\usepackage{bm}
\usepackage{amsthm}

\allowdisplaybreaks
\usepackage{graphicx}
\usepackage{booktabs}
\usepackage{array}
\usepackage{multirow}

\usepackage{xcolor}
\usepackage{cite}

\usepackage[hidelinks]{hyperref}

\usepackage[
    capitalise,
    nameinlink,
    noabbrev
]{cleveref}

\theoremstyle{plain}
\newtheorem{theorem}{Theorem}
\newtheorem{proposition}{Proposition}

\newtheorem{corollary}{Corollary}

\theoremstyle{definition}

\newtheorem{assumption}{Assumption}

\theoremstyle{remark}

\crefname{assumption}{Assumption}{Assumptions}
\Crefname{assumption}{Assumption}{Assumptions}

\crefname{definition}{Definition}{Definitions}
\Crefname{definition}{Definition}{Definitions}

\crefname{theorem}{Theorem}{Theorems}
\Crefname{theorem}{Theorem}{Theorems}

\crefname{proposition}{Proposition}{Propositions}
\Crefname{proposition}{Proposition}{Propositions}

\crefname{lemma}{Lemma}{Lemmas}
\Crefname{lemma}{Lemma}{Lemmas}

\crefname{corollary}{Corollary}{Corollaries}
\Crefname{corollary}{Corollary}{Corollaries}

\crefname{remark}{Remark}{Remarks}
\Crefname{remark}{Remark}{Remarks}

\crefname{example}{Example}{Examples}
\Crefname{example}{Example}{Examples}

\crefname{equation}{Eq.}{Eqs.}
\Crefname{equation}{Equation}{Equations}

\crefname{section}{Section}{Sections}
\Crefname{section}{Section}{Sections}

\crefname{subsection}{Section}{Sections}
\Crefname{subsection}{Section}{Sections}

\crefname{figure}{Fig.}{Figs.}
\Crefname{figure}{Figure}{Figures}

\crefname{table}{Table}{Tables}
\Crefname{table}{Table}{Tables}

\usepackage{algorithm}
\usepackage{algpseudocode}

\begin{document}

\title{Reusable Neural Reachability via Local Refinement}


\author{
Prashant Solanki, Jasper Van Beers, Coen De Visser
}

\maketitle

\begin{abstract}
Repeated obstacle geometry enables reuse of collision-avoidance computations,
but individually avoidable obstacles may require incompatible controls.
We develop a modular Hamilton Jacobi reachability framework using bounded,
saturated finite-horizon reach costs. For translation-invariant dynamics,
singleton values translate exactly, and their pointwise minimum upper-bounds
the clutter value, yielding an inner approximation of unavoidable collision.
Localized continuation problems tighten this composition without changing the
exterior; they recover the exact value where all admissible trajectories remain
inside the refinement region. A neural realization freezes the singleton and
learns masked, nonnegative interaction corrections through Bellman updates.
In Dubins-car experiments, refinement reduces inside-patch value RMSE by
46.3\% and near-zero RMSE by 58.5\%. Across six relocated configurations in
300--1000~\(\mathrm{m}^2\) workspaces, the frozen library achieves approximately
63--73\% lower interaction RMSE than the mean global neural warm-start baseline,
with no additional training after relocation versus 10,000 adaptation updates.
This comparison concerns repeated pair geometry under the reported training
recipe. A PiCar-X demonstration combines online singleton reuse with PID
waypoint following and a greedy look-ahead filter. 
\end{abstract}

\begin{IEEEkeywords}
Hamilton Jacobi reachability, collision avoidance, modular learning,
robot safety, value function approximation.
\end{IEEEkeywords}

\section{Introduction}
\label{sec:introduction}

Consider a car-like robot operating among obstacles whose shapes recur but whose
positions change between deployments. Relocating an obstacle changes the
collision-avoidance problem even when the robot dynamics and obstacle geometry
remain unchanged. For a vehicle with limited turning authority, whether
collision is avoidable depends not only on position, but also on heading,
admissible controls, and the available maneuvering time. Hamilton Jacobi (HJ)
reachability captures these dynamical constraints through a value function and
its associated reachable set~\cite{mitchell2005time}. The computational question
addressed here is how to retain and reuse this information across clutter layouts
without repeatedly learning an entire scene-dependent value function.

Repeated geometry creates reusable structure, but obstacle interactions create
information that is not present in isolated obstacle problems. A maneuver that
avoids one obstacle may be incompatible with the maneuver required to avoid
another. Consequently, accurate singleton values do not in general determine
the value for their union. This paper develops a modular finite-horizon
reach-cost framework that separates these two effects. Translation symmetry
preserves and reuses the singleton component, while local refinement represents
the interaction-dependent information missing from minimum composition. In the
neural realization, a singleton network is trained once and frozen, and
spatially localized residual networks supply additional correction. The guiding
principle is to learn reusable structure once and allocate further capacity to
the local interactions that composition does not capture.

\subsection{Related Work}
\label{subsec:intro_related_work}

\paragraph{Reachability computation and neural approximation}
HJ reachability characterizes reachable sets through optimal-control or
differential-game value functions and their partial differential
equations~\cite{mitchell2005time}. For suitable dynamics, decomposition reduces
the dimension of the required calculations~\cite{chen2018decomposition}.
Neural approximation offers a complementary representation: DeepReach
\cite{bansal2021deepreach} learns HJ value functions through their differential
characterization using sinusoidal networks related to
SIRENs~\cite{sitzmann2020siren}. Bellman formulations connect reachability to
sample-based learning. Fisac et al.~\cite{fisac2019bridging} develop a discounted
safety Bellman formulation, Hsu et al.~\cite{hsu2021safety} address reach-avoid
reinforcement learning, and El-Hajj et al.~\cite{elhajj2026reachcost} connect
discounted reach costs to non-additive Bellman operators. These methods provide
tools for learning values; the present work concerns retaining reusable
information when obstacles are rearranged.

\paragraph{Warm starting and local value refinement}
Herbert et al.~\cite{herbert2019initializations} initialize HJ computations with
previously available values, establishing conservative results and conditions
for exact recovery. Tonkens and Herbert~\cite{tonkens2022refining} introduce
\emph{refineCBF}, which refines candidate control barrier functions through
reachability-based dynamic programming. HJ-Patch~\cite{tonkens2024hjpatch}
concentrates updates on selected states near potentially unsafe portions of an
approximate safety boundary. Tonkens et al.~\cite{tonkens2026onthefly} further
develop in-the-loop refinement, including the joint feasibility of
obstacle-dependent functions and physical ground-vehicle and quadcopter
experiments. These works establish the utility of warm initialization and local
repair. In our neural evaluation, weight warm starting provides an optimization
initialization, while freezing and localization explicitly preserve previously
represented components. We compare these forms of reuse under a fixed
adaptation budget.

\paragraph{Reuse through symmetry and environmental conditioning}
Funnel libraries reuse robust feedback tubes~\cite{majumdar2017funnels}, and
SceneChecker exploits symmetry in scenario verification~\cite{sibai2021scenechecker}.
Parameter-conditioned reachability supports online queries across anticipated
system and environment changes~\cite{borquez2023parameter}. Domain-adaptive
safety filters~\cite{manda2024domainadaptive} and HJRNO~\cite{li2025hjrno} use
operator learning to represent families of safety or reachability functions.
CN-CBF~\cite{derajic2026cncbf} composes HJ-derived neural robot--obstacle
functions evaluated in relative coordinates. These approaches motivate
representations that retain useful information across environments. Here,
translation acts directly on a fixed singleton value, and separately retained
local modules represent corrections to the resulting clutter composition.

\begin{figure}[t]
    \centering
    \includegraphics[width=0.95\linewidth]{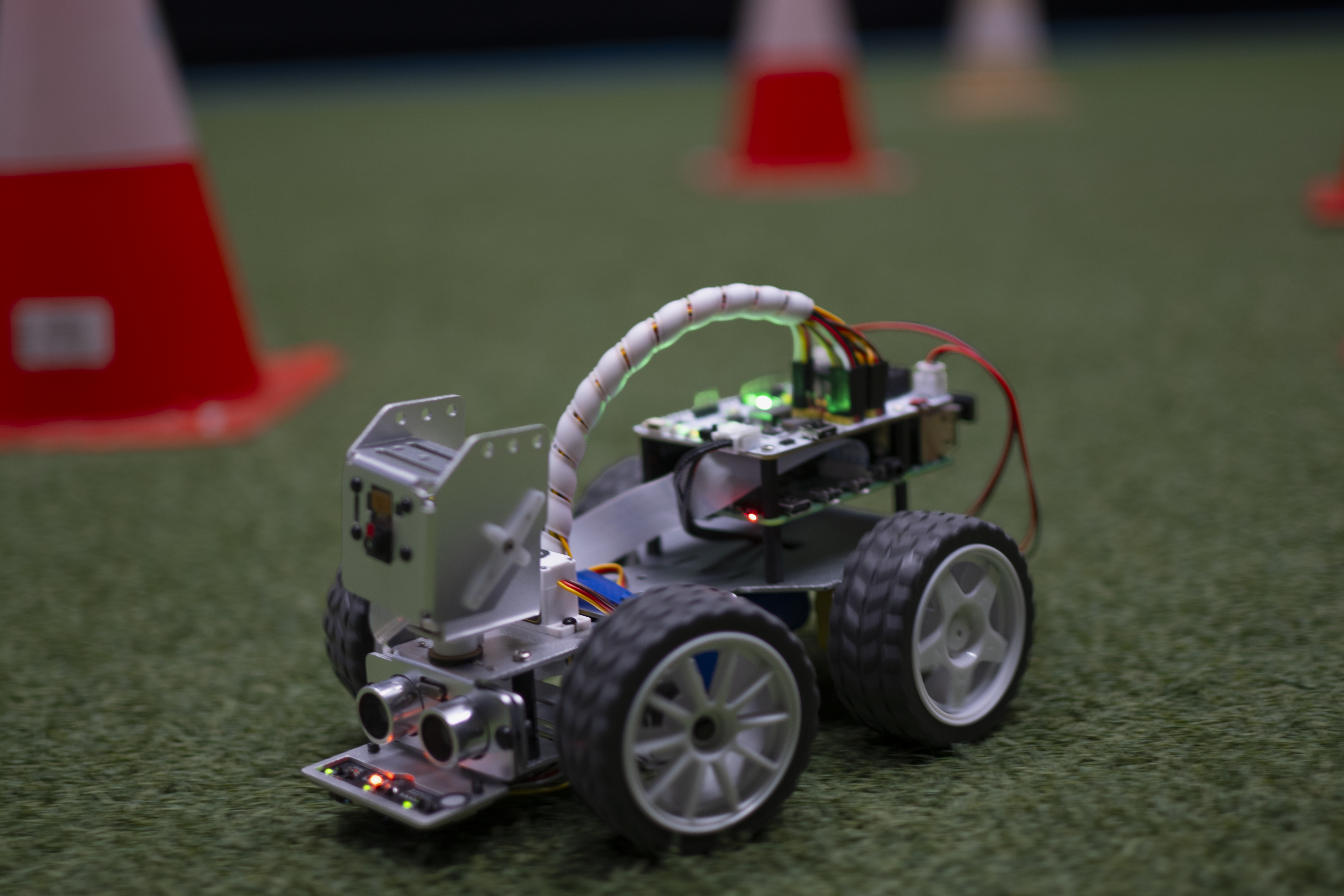}
    \caption{SunFounder PiCar-X platform used for the physical
    feasibility experiments. Vehicle pose is measured using an
    external OptiTrack motion-capture system, while learned
    singleton reachability values are reused for obstacle avoidance.}
    \label{fig:picarx_platform}
\end{figure}

\subsection{Overview and Contributions}
\label{subsec:intro_contributions}

For translation-invariant dynamics, translating an obstacle by \(Ec_i\)
translates its singleton value:
\(W_i(\tau,x)=W_0(\tau,x-Ec_i)\), where \(\tau\) is the remaining horizon
and \(E\) embeds the positional displacement into the state. Nevertheless,
writing \(J_i(\tau,x;u)\) for a singleton trajectory payoff and
\(\mathbb U_\tau\) for admissible controls on \([0,\tau]\),
\begin{equation}
\begin{split}
    W_{\mathcal C}(\tau,x)
    &=\sup_{u\in\mathbb U_\tau}\min_i J_i(\tau,x;u)\\
    &\leq\min_i\sup_{u\in\mathbb U_\tau}J_i(\tau,x;u)
      =U_{\mathrm{sg}}(\tau,x).
\end{split}
    \label{eq:intro_common_control_gap}
\end{equation}
The composition \(U_{\mathrm{sg}}\) permits different control choices for
different obstacles, whereas \(W_{\mathcal C}\) requires one common control.
This gap motivates local continuation problems: the clutter reach cost is used
until a trajectory exits a selected region or exhausts its horizon, with the
current upper value supplying continuation data on exit. Taking the minimum of
the current value and this local result tightens the upper bound and preserves
the exterior. The neural architecture follows this separation by freezing the
singleton network and adding localized residual modules. The singleton network
size does not grow with workspace area.

Function-level reuse and neural weight warm starting address different
questions. Copying a global network's weights provides an initialization, but
subsequent optimization can change its represented function throughout the
workspace. Translating frozen modules instead retains the represented
singleton and local-correction functions without updating their parameters.
Our experiments compare these alternatives after the same scene edits.
Under the reported recipe, the modular library retains lower local error than
the adapted global baselines, including the warm-start model; warm starting
still improves early optimization relative to global training from scratch.
This is an empirical comparison of representations and adaptation procedures,
not a claim that all possible warm-start algorithms are inferior.

The contributions are as follows.
\begin{enumerate}
    \item \textbf{Exact reuse and optimistic composition.}
    We establish translation of finite-horizon singleton values and avoid sets,
    an upper-bound-preserving finite-template extension, and minimum
    composition for clutter. The untruncated exact composition is a Bellman
    supersolution, and its nonpositive level set is an inner approximation of
    the clutter avoid set.

    \item \textbf{Localized interaction refinement.}
    We construct continuation and patch operators satisfying
    \(W_{\mathcal C}\leq\mathcal P_{\mathcal D}U\leq U\) for a valid upper
    value \(U\), with no change outside \(\mathcal D\). Exact refinements
    give nested avoid-set inner approximations and recover the full value on
    states whose finite-horizon reachable tubes remain inside the region.

    \item \textbf{Frozen neural modules and reuse evaluation.}
    We realize the construction with a frozen singleton and localized residual
    networks trained through reach-cost Bellman updates. Dubins-car studies
    separate singleton error from composition error, evaluate interaction
    refinement, and reuse the same components after translation into larger
    workspaces. Across all six relocated cases, the library has approximately
    63--73\% lower interaction RMSE than the mean neural weight-warm-start
    baseline, without relocation-specific training. Global scratch and
    warm-start learning are compared under the stated architecture and
    update budget. A PiCar-X demonstration implements
    singleton reuse with a PID nominal controller and a one-step safety filter.
\end{enumerate}

The ordering and exactness guarantees concern exact values and operators.
Neural accuracy is evaluated against independent numerical references, and the
hardware study establishes practical feasibility. Section~\ref{sec:background}
defines the reach-cost problem; Sections~\ref{sec:modular_value} and
\ref{sec:localized_refinement} develop reuse, composition, and refinement.
Sections~\ref{sec:experiments} and~\ref{sec:results} present the numerical and
physical evaluations, followed by the conclusion.

\section{Background}
\label{sec:background}

\subsection{Controlled System and Translation Symmetry}
\label{subsec:controlled_system}

Consider the autonomous controlled system
\begin{equation}
    \dot x(t)=f\bigl(x(t),u(t)\bigr),
    \qquad x(0)=x_0\in\mathbb R^n,
    \label{eq:controlled_system}
\end{equation}
with controls in a set \(\mathcal U\subset\mathbb R^m\). For each
finite horizon \(\tau\geq0\), the admissible controls are
\begin{equation}
    \mathbb U_\tau:=
    \{u:[0,\tau]\to\mathcal U\mid u\text{ is Lebesgue measurable}\}.
    \label{eq:admissible_controls_finite}
\end{equation}

\begin{assumption}[Controlled dynamics]
\label{ass:controlled_dynamics}
The set \(\mathcal U\) is nonempty and compact. The map
\(f:\mathbb R^n\times\mathcal U\to\mathbb R^n\) is continuous and
locally Lipschitz in the state uniformly over the control: for every
\(R>0\), there is \(L_R<\infty\) such that, for
\(x,y\in\overline B_R\) and \(u\in\mathcal U\),
\begin{equation}
    \|f(x,u)-f(y,u)\|\leq L_R\|x-y\|.
    \label{eq:local_lipschitz_dynamics}
\end{equation}
There is also \(C_f>0\) such that
\begin{equation}
    \|f(x,u)\|\leq C_f(1+\|x\|),
    \qquad (x,u)\in\mathbb R^n\times\mathcal U.
    \label{eq:linear_growth_dynamics}
\end{equation}
\end{assumption}

These conditions give a unique Carath\'eodory trajectory
\(\phi_x^u\) on the entire domain of each admissible control; the
linear-growth bound rules out finite escape. Let
\(E\in\mathbb R^{n\times d}\), with \(\operatorname{rank}(E)=d\),
embed translations in the \(d\) coordinates along which obstacles
may be relocated.

\begin{assumption}[Translation invariance]
\label{ass:translation_invariance}
For every \(x\in\mathbb R^n\), \(u\in\mathcal U\), and
\(c\in\mathbb R^d\),
\begin{equation}
    f(x+Ec,u)=f(x,u).
    \label{eq:translation_invariant_dynamics}
\end{equation}
\end{assumption}

Uniqueness then gives the flow identity
\begin{equation}
    \phi_{x+Ec}^u(t)=\phi_x^u(t)+Ec.
    \label{eq:translation_equivariant_flow}
\end{equation}
Thus, applying the same control to translated initial states
translates the complete trajectory.

\subsection{Template Obstacle and Reach Cost}
\label{subsec:template_obstacle}

Let \(\mathcal O_0\subset\mathbb R^n\) be a closed template obstacle,
represented by a reach cost \(g_0:\mathbb R^n\to\mathbb R\).

\begin{assumption}[Template reach cost]
\label{ass:template_reach_cost}
The function \(g_0\) is locally Lipschitz and satisfies, for constants
\(\underline g,\overline g>0\),
\begin{equation}
    -\underline g\leq g_0(x)\leq\overline g,
    \label{eq:reach_cost_bounds}
\end{equation}
with the sign convention
\begin{equation}
    g_0(x)
    \begin{cases}
        <0,&x\in\operatorname{int}(\mathcal O_0),\\
        =0,&x\in\partial\mathcal O_0,\\
        >0,&x\notin\mathcal O_0.
    \end{cases}
    \label{eq:reach_cost_sign}
\end{equation}
For some \(R_g>0\), it also satisfies
\begin{equation}
    g_0(x)=\overline g
    \quad\text{if }\operatorname{dist}(x,\mathcal O_0)\geq R_g.
    \label{eq:reach_cost_saturation}
\end{equation}
\end{assumption}

Hence \(\mathcal O_0=\{x:g_0(x)\leq0\}\): contact with the obstacle
boundary is unsafe, and the positive far-field cost is saturated.

\subsection{Translated Obstacles and Clutter}
\label{subsec:translated_clutter}

A finite configuration \(\mathcal C=(c_1,\ldots,c_N)\), with
\(N\geq1\) and \(c_i\in\mathbb R^d\), defines
\begin{align}
    \mathcal O_i&:=\mathcal O_0+Ec_i,
    &g_i(x)&:=g_0(x-Ec_i),
    \label{eq:translated_reach_cost}\\
    \mathcal O_{\mathcal C}&:=\bigcup_{i=1}^N\mathcal O_i,
    &g_{\mathcal C}(x)&:=\min_{1\leq i\leq N}g_i(x).
    \label{eq:clutter_reach_cost}
\end{align}
Thus \(\mathcal O_i=\{g_i\leq0\}\) and
\(\mathcal O_{\mathcal C}=\{g_{\mathcal C}\leq0\}\).
The finite minimum preserves local Lipschitz continuity and the
bounds in \eqref{eq:reach_cost_bounds}. The obstacles may overlap.

\subsection{Finite-Horizon Reach-Cost Value and Avoid Set}
\label{subsec:reach_cost_value}

Fix a maximum horizon \(T>0\) and a discount rate \(\lambda\geq0\).
For \(\tau\in[0,T]\), \(x\in\mathbb R^n\), and
\(u\in\mathbb U_\tau\), define
\begin{align}
    J_{\mathcal C}(\tau,x;u)
    &:=\min_{t\in[0,\tau]}e^{-\lambda t}
        g_{\mathcal C}\bigl(\phi_x^u(t)\bigr),
    \label{eq:finite_reach_payoff}\\
    W_{\mathcal C}(\tau,x)
    &:=\sup_{u\in\mathbb U_\tau}J_{\mathcal C}(\tau,x;u).
    \label{eq:finite_reach_value}
\end{align}
The controller maximizes its minimum discounted reach cost over the
remaining horizon. The corresponding unavoidable-collision set is
\begin{equation}
\begin{split}
    \mathcal A_{\mathcal C}(\tau):=
    \{x\in\mathbb R^n\mid{}&\forall u\in\mathbb U_\tau,\
       \exists t\in[0,\tau]:\\[-1mm]
       &\phi_x^u(t)\in\mathcal O_{\mathcal C}\}.
\end{split}
    \label{eq:finite_avoid_set}
\end{equation}
Continuity on the compact time interval ensures that the time
minimum is attained. For each control, that minimum is nonpositive
exactly when its trajectory meets the obstacle. Therefore
\begin{equation}
    \mathcal A_{\mathcal C}(\tau)=\{x:W_{\mathcal C}(\tau,x)\leq0\}.
    \label{eq:finite_sign_characterization}
\end{equation}
Its complement is \(\{W_{\mathcal C}>0\}\): a collision-free control
has a strictly positive minimum on \([0,\tau]\).
Neither conclusion requires attainment of the control supremum.
Replacing \(g_{\mathcal C}\) by \(g_0\) or \(g_i\) defines
\(W_0,W_i\) and \(\mathcal A_0,\mathcal A_i\), respectively.

The cost bounds and evaluation of the trajectory minimum at \(t=\tau\) give
\begin{equation}
    -\underline g\leq W_{\mathcal C}(\tau,x)
    \leq b_\lambda(\tau):=\overline g e^{-\lambda\tau}.
    \label{eq:finite_horizon_upper_baseline}
\end{equation}
The same bounds hold for every singleton. The positive quantity
\(b_\lambda\) is an upper baseline; saturation of the cost at the
initial state alone does not imply that the value equals it.

\subsection{Dynamic Programming and Hamilton-Jacobi Characterization}
\label{subsec:dpp_hj}

For \(h\in[0,\tau]\) and \(u\in\mathbb U_h\), let
\begin{equation}
    S_{\mathcal C}(h,x;u):=
    \min_{t\in[0,h]}e^{-\lambda t}g_{\mathcal C}\bigl(\phi_x^u(t)\bigr).
    \label{eq:short_horizon_stopping_value}
\end{equation}
Splitting the trajectory at \(h\) gives the dynamic programming
principle
\begin{equation}
\begin{split}
    W_{\mathcal C}(\tau,x)=\sup_{u\in\mathbb U_h}\min\bigl\{
        &S_{\mathcal C}(h,x;u),\\
        &e^{-\lambda h}W_{\mathcal C}
          \bigl(\tau-h,\phi_x^u(h)\bigr)\bigr\}.
\end{split}
    \label{eq:dpp_reach_cost}
\end{equation}
Fix a Bellman step \(\sigma\in(0,T]\), and set
\(h_\tau:=\min\{\sigma,\tau\}\). For bounded candidates
\(\Psi\) satisfying \(\Psi(0,x)=g_{\mathcal C}(x)\), define
\begin{equation}
\begin{split}
    (\mathcal T_{\mathcal C}\Psi)(\tau,x)
    :=\sup_{u\in\mathbb U_{h_\tau}}\min\bigl\{
        &S_{\mathcal C}(h_\tau,x;u),\\
        &e^{-\lambda h_\tau}\Psi
          \bigl(\tau-h_\tau,\phi_x^u(h_\tau)\bigr)\bigr\}.
\end{split}
    \label{eq:clutter_bellman_operator}
\end{equation}
This operator is order preserving and fixes the reach-cost value:
\begin{gather}
    \Psi_1\leq\Psi_2\ \Longrightarrow\
       \mathcal T_{\mathcal C}\Psi_1\leq\mathcal T_{\mathcal C}\Psi_2,
    \label{eq:bellman_monotonicity}\\
    W_{\mathcal C}=\mathcal T_{\mathcal C}W_{\mathcal C}.
    \label{eq:bellman_fixed_point}
\end{gather}

With \(H(x,p):=\sup_{u\in\mathcal U}p^\top f(x,u)\), the dynamic
programming principle also yields the Hamilton-Jacobi variational
inequality, satisfied in the viscosity sense,
\begin{equation}
    \max\{W_{\mathcal C}-g_{\mathcal C},\
       \partial_\tau W_{\mathcal C}
       -H(x,\nabla_xW_{\mathcal C})+\lambda W_{\mathcal C}\}=0,
    \label{eq:reach_cost_hjvi}
\end{equation}
with \(W_{\mathcal C}(0,x)=g_{\mathcal C}(x)\).
The Bellman and HJ formulations apply equally to the template and
singletons. The results below use the value definition and dynamic
programming directly; they do not require a particular numerical
solver or an additional PDE uniqueness assertion.

\section{Translation-Based Modular Value Construction}
\label{sec:modular_value}

\subsection{Exact Translation of Singleton Values}
\label{subsec:exact_translation}

Translation invariance identifies the part of the reachability
problem that can be reused without solving a new control problem.

\begin{theorem}[Exact translation of singleton values and avoid sets]
\label{thm:exact_value_translation}
Under Assumptions~\ref{ass:controlled_dynamics}--\ref{ass:template_reach_cost},
for every \(i\in\{1,\ldots,N\}\), \(\tau\in[0,T]\), and
\(x\in\mathbb R^n\),
\begin{align}
    W_i(\tau,x)&=W_0(\tau,x-Ec_i),
    \label{eq:exact_singleton_translation}\\
    \mathcal A_i(\tau)&=Ec_i+\mathcal A_0(\tau).
    \label{eq:translated_avoid_set}
\end{align}
\end{theorem}

\begin{proof}
For each \(u\in\mathbb U_\tau\), the translated reach cost and
flow identity \eqref{eq:translation_equivariant_flow} give
\begin{equation}
\begin{split}
    J_i(\tau,x;u)
    &=\min_{t\in[0,\tau]}e^{-\lambda t}
       g_0\bigl(\phi_x^u(t)-Ec_i\bigr)\\
    &=J_0(\tau,x-Ec_i;u).
\end{split}
    \label{eq:translated_payoff_identity}
\end{equation}
Taking the supremum over the identical control domains proves the
value identity. Equation~\eqref{eq:finite_sign_characterization}
then gives the set identity.
\end{proof}

The identity concerns exact values. Applying the same coordinate
shift to a numerical or learned template reuses its represented
function without further training, while retaining its approximation
error.

\subsection{Finite Template Representation}
\label{subsec:finite_template}

A finite representation window requires a rule for queries outside
that window. Let \(\Omega_0\subset\mathbb R^n\) be the template
representation region; it may be unbounded in non-translated
coordinates. Periodic coordinates can instead be stored over one
period, as for the Dubins heading in
Section~\ref{subsec:experimental_setup}.

\begin{assumption}[Template representation region]
\label{ass:template_region}
The region contains the non-saturated cost in its interior:
\begin{equation}
    \{x:g_0(x)<\overline g\}\subset\operatorname{int}(\Omega_0).
    \label{eq:template_region_condition}
\end{equation}
\end{assumption}

Thus \(g_0=\overline g\) outside \(\operatorname{int}(\Omega_0)\).
Using the upper baseline from
\eqref{eq:finite_horizon_upper_baseline}, extend the restricted
full-space value by
\begin{equation}
    \widetilde W_0(\tau,x):=
    \begin{cases}
        W_0(\tau,x),&x\in\Omega_0,\\
        b_\lambda(\tau),&x\notin\Omega_0.
    \end{cases}
    \label{eq:extended_template_value}
\end{equation}

\begin{proposition}[Template-extension upper bound]
\label{prop:template_extension}
On \([0,T]\times\mathbb R^n\),
\begin{equation}
    W_0\leq\widetilde W_0\leq b_\lambda,
    \qquad \widetilde W_0(0,x)=g_0(x).
    \label{eq:template_extension_upper_bound}
\end{equation}
\end{proposition}

\begin{proof}
Inside \(\Omega_0\), the bound follows from
\eqref{eq:finite_horizon_upper_baseline}; outside, the extension equals
that upper bound. At \(\tau=0\), it agrees with \(g_0\) both inside
the region and on the exterior plateau.
\end{proof}

Cost saturation on \(\partial\Omega_0\) does not ensure value
saturation there: a trajectory can subsequently enter the
non-saturated region. The extension can therefore introduce an
optimistic upward jump, but the proposition remains valid. A
sufficient condition for exact extension over \([0,T]\) is that
every trajectory starting outside \(\Omega_0\) remain on the cost
plateau until \(T\); then \(W_0=b_\lambda\) there. More generally,
the proposition concerns the restriction of the full-space value.
Computing that restriction on a finite numerical domain additionally
requires suitable boundary data or padding covering the relevant
trajectories.

\subsection{Modular Clutter Initialization}
\label{subsec:modular_clutter_initialization}

For each obstacle, let
\begin{equation}
    \widetilde W_i(\tau,x):=\widetilde W_0(\tau,x-Ec_i).
    \label{eq:translated_extended_value}
\end{equation}
Theorem~\ref{thm:exact_value_translation} and
Proposition~\ref{prop:template_extension} imply
\(W_i\leq\widetilde W_i\). Define the exact singleton composition
and finite-template initialization, respectively, by
\begin{align}
    U_{\mathrm{sg}}(\tau,x)&:=\min_{1\leq i\leq N}W_i(\tau,x),
    \label{eq:exact_singleton_composition}\\
    U_{\mathcal C}^{(0)}(\tau,x)&:=
       \min_{1\leq i\leq N}\widetilde W_i(\tau,x).
    \label{eq:modular_clutter_initialization}
\end{align}

\begin{theorem}[Modular clutter upper bound]
\label{thm:modular_clutter_upper_bound}
For every \((\tau,x)\in[0,T]\times\mathbb R^n\),
\begin{equation}
    W_{\mathcal C}(\tau,x)\leq U_{\mathrm{sg}}(\tau,x)
    \leq U_{\mathcal C}^{(0)}(\tau,x)\leq b_\lambda(\tau),
    \label{eq:modular_upper_bound_hierarchy}
\end{equation}
and \(U_{\mathcal C}^{(0)}(0,x)=g_{\mathcal C}(x)\).
\end{theorem}

\begin{proof}
Since \(g_{\mathcal C}\leq g_i\), every control satisfies
\(J_{\mathcal C}\leq J_i\), so \(W_{\mathcal C}\leq W_i\) for
each \(i\). Taking their minimum gives the first inequality; the
extension bounds give the others. At zero horizon,
\(U_{\mathcal C}^{(0)}(0,x)=\min_i g_i(x)=g_{\mathcal C}(x)\).
\end{proof}

The finite minimum is nonpositive precisely when at least one
singleton value is nonpositive. Consequently,
\begin{equation}
    \{x:U_{\mathrm{sg}}(\tau,x)\leq0\}
    =\bigcup_{i=1}^N\mathcal A_i(\tau)
    \subseteq\mathcal A_{\mathcal C}(\tau).
    \label{eq:singleton_avoid_union}
\end{equation}
For the finite-template certificate
\(\mathcal A_{\mathcal C,\mathrm{cert}}^{(0)}(\tau)
:=\{x:U_{\mathcal C}^{(0)}(\tau,x)\leq0\}\),
\begin{equation}
    \mathcal A_{\mathcal C,\mathrm{cert}}^{(0)}(\tau)
    \subseteq\bigcup_{i=1}^N\mathcal A_i(\tau)
    \subseteq\mathcal A_{\mathcal C}(\tau).
    \label{eq:initial_certificate_inclusion}
\end{equation}
An optimistic template extension may shrink this certificate; it
equals the singleton union when the extension is exact. This is a
certificate of unavoidable collision: positivity of an upper value
does not certify safety.

\begin{proposition}[Stability of approximate composition]
\label{prop:approximate_composition_stability}
Suppose a template approximation \(\widehat W_0\) satisfies
\(\sup_{[0,T]\times\Omega_0}|\widehat W_0-W_0|\leq\varepsilon\).
Extend it by the same exterior baseline, and let
\(\widehat U_{\mathcal C}^{(0)}\) be the minimum of its translated
copies. Then
\begin{equation}
    \sup_{[0,T]\times\mathbb R^n}
       |\widehat U_{\mathcal C}^{(0)}-U_{\mathcal C}^{(0)}|
    \leq\varepsilon.
    \label{eq:composition_approximation_stability}
\end{equation}
\end{proposition}

\begin{proof}
The extended templates have zero difference outside \(\Omega_0\)
and difference at most \(\varepsilon\) inside. Translation preserves
this bound, and \(|\min_i a_i-\min_i b_i|\leq\max_i|a_i-b_i|\).
\end{proof}

Thus composition does not amplify a uniform template error with
obstacle count. This conditional bound does not provide a neural
certificate unless \(\varepsilon\) is independently established;
empirical RMSE is not such a uniform bound.

\subsection{Common-Control Gap and Bellman Refinement Indicator}
\label{subsec:common_control_gap}

For a fixed control, finite obstacle and time minima commute:
\begin{equation}
    J_{\mathcal C}(\tau,x;u)=\min_i J_i(\tau,x;u).
    \label{eq:clutter_payoff_min_singletons}
\end{equation}
Hence the clutter and singleton-composed values differ in the
order of optimization:
\begin{align}
    W_{\mathcal C}(\tau,x)
       &=\sup_{u\in\mathbb U_\tau}\min_i J_i(\tau,x;u),
    \label{eq:true_common_control_value}\\
    U_{\mathrm{sg}}(\tau,x)
       &=\min_i\sup_{u\in\mathbb U_\tau}J_i(\tau,x;u).
    \label{eq:singleton_independent_control_value}
\end{align}
The first requires one control for all obstacles; the second allows
each obstacle its own control. The max--min inequality therefore
explains the optimism in
\eqref{eq:modular_upper_bound_hierarchy}. Translation preserves
singleton information exactly, but cannot supply this
configuration-dependent common-control information.

\begin{proposition}[Singleton composition is a Bellman supersolution]
\label{prop:singleton_bellman_supersolution}
The exact singleton composition satisfies
\begin{equation}
    \mathcal T_{\mathcal C}U_{\mathrm{sg}}\leq U_{\mathrm{sg}}
    \label{eq:singleton_bellman_supersolution}
\end{equation}
on \([0,T]\times\mathbb R^n\).
\end{proposition}

\begin{proof}
Fix \((\tau,x)\) and set \(h=h_\tau\). Let \(S_i\) be the
stopping value \eqref{eq:short_horizon_stopping_value} with cost
\(g_i\), and define
\begin{equation}
\begin{split}
    Q_i(\tau,x;u):=\min\{&S_i(h,x;u),\\
       &e^{-\lambda h}W_i(\tau-h,\phi_x^u(h))\}.
\end{split}
    \label{eq:singleton_bellman_branch}
\end{equation}
Commuting finite minima and applying each singleton's DPP gives
\begin{align*}
    (\mathcal T_{\mathcal C}U_{\mathrm{sg}})(\tau,x)
       &=\sup_{u\in\mathbb U_h}\min_i Q_i(\tau,x;u)\\
       &\leq\min_i\sup_{u\in\mathbb U_h}Q_i(\tau,x;u)\\
       &=\min_iW_i(\tau,x)=U_{\mathrm{sg}}(\tau,x).
\end{align*}
\end{proof}

Define the one-step Bellman residual by
\begin{equation}
    r_{\mathcal C}[U]:=U-\mathcal T_{\mathcal C}U.
    \label{eq:clutter_bellman_residual}
\end{equation}
For exact singleton composition, \(r_{\mathcal C}[U_{\mathrm{sg}}]\geq0\).
A positive residual identifies a downward correction obtained by
enforcing a common control over the Bellman step. A zero residual
at a state does not imply an exact clutter value: later conflicts
may require repeated backups to propagate to that state.

The finite-template initialization remains an upper bound, but its
extension need not preserve the supersolution property. For any
bounded \(U\geq W_{\mathcal C}\) with
\(U(0,x)=g_{\mathcal C}(x)\), monotonicity instead gives
\(W_{\mathcal C}=\mathcal T_{\mathcal C}W_{\mathcal C}
\leq\mathcal T_{\mathcal C}U\), and therefore
\begin{equation}
    W_{\mathcal C}\leq\min\{U,\mathcal T_{\mathcal C}U\}\leq U.
    \label{eq:safe_bellman_correction}
\end{equation}
Equivalently, with \([z]_+:=\max\{z,0\}\), define
\(\Delta_{\mathcal C}[U]:=[r_{\mathcal C}[U]]_+\). Then
\begin{equation}
    \min\{U,\mathcal T_{\mathcal C}U\}=U-\Delta_{\mathcal C}[U].
    \label{eq:bellman_correction_residual_form}
\end{equation}
The next section localizes this downward correction while retaining
the current representation outside the selected region.

\section{Localized Value Refinement}
\label{sec:localized_refinement}

\subsection{Local Continuation and Upper-Bound Preservation}
\label{subsec:localized_exact_refinement}

Let \(\mathcal D\subset\mathbb R^n\) be an open refinement region
and let \(U\) be a bounded current value satisfying
\begin{equation}
    W_{\mathcal C}(\tau,x)\leq U(\tau,x),
    \qquad (\tau,x)\in[0,T]\times\mathbb R^n.
    \label{eq:current_global_upper_bound}
\end{equation}
When preserving the zero-horizon condition, we also assume
\begin{equation}
    U(0,x)=g_{\mathcal C}(x).
    \label{eq:current_upper_boundary_condition}
\end{equation}
The region may be unbounded in coordinates that are not localized.
For \(x\in\mathcal D\) and \(u\in\mathbb U_\tau\), define the
first exit time, capped by the remaining horizon, as
\begin{equation}
    \eta_{\mathcal D}^{u}(\tau,x)
    :=\tau\wedge\inf\{t\in[0,\tau]:
          \phi_x^u(t)\notin\mathcal D\},
    \label{eq:local_exit_time}
\end{equation}
with \(\inf\varnothing=+\infty\). Writing
\(\eta=\eta_{\mathcal D}^{u}(\tau,x)\), let
\begin{equation}
    S_{\mathcal C}^{\mathcal D}(\tau,x;u)
    :=\min_{t\in[0,\eta]}e^{-\lambda t}
          g_{\mathcal C}(\phi_x^u(t)).
    \label{eq:local_stopping_cost}
\end{equation}
The local continuation operator is
\begin{equation}
\begin{split}
    (\mathcal L_{\mathcal D}U)(\tau,x)
    :=\sup_{u\in\mathbb U_\tau}\min\bigl\{
       &S_{\mathcal C}^{\mathcal D}(\tau,x;u),\\
       &e^{-\lambda\eta}U(\tau-\eta,\phi_x^u(\eta))\bigr\}.
\end{split}
    \label{eq:local_continuation_operator}
\end{equation}
Thus the local problem evaluates the clutter cost until exit or
horizon completion and inherits \(U\) as the continuation value.
Accept only a decrease inside the selected region:
\begin{equation}
    (\mathcal P_{\mathcal D}U)(\tau,x):=
    \begin{cases}
      \min\{U(\tau,x),(\mathcal L_{\mathcal D}U)(\tau,x)\},
        &x\in\mathcal D,\\
      U(\tau,x),&x\notin\mathcal D.
    \end{cases}
    \label{eq:local_patch_operator}
\end{equation}

\begin{theorem}[Upper-bound-preserving local refinement]
\label{prop:local_upper_bound_preservation}
\label{thm:certified_local_refinement}
If \(U\) satisfies \eqref{eq:current_global_upper_bound}, then
\begin{align}
    W_{\mathcal C}(\tau,x)&\leq(\mathcal L_{\mathcal D}U)(\tau,x),
          &&x\in\mathcal D,
    \label{eq:local_operator_upper_bound}\\
    W_{\mathcal C}&\leq\mathcal P_{\mathcal D}U\leq U,
    \label{eq:local_refinement_ordering}\\
    (\mathcal P_{\mathcal D}U)(\tau,x)&=U(\tau,x),
          &&x\notin\mathcal D.
    \label{eq:local_refinement_unchanged_outside}
\end{align}
If \eqref{eq:current_upper_boundary_condition} also holds, then
\((\mathcal P_{\mathcal D}U)(0,x)=g_{\mathcal C}(x)\).
\end{theorem}

\begin{proof}
Fix \((\tau,x)\) with \(x\in\mathcal D\) and a control
\(u\in\mathbb U_\tau\). Set \(y=\phi_x^u(\eta)\) and
\(u_\eta(r)=u(\eta+r)\) on \([0,\tau-\eta]\). Splitting this
control's payoff at its exit time gives
\begin{equation}
\begin{split}
    J_{\mathcal C}(\tau,x;u)=\min\bigl\{
       &S_{\mathcal C}^{\mathcal D}(\tau,x;u),\\
       &e^{-\lambda\eta}J_{\mathcal C}(\tau-\eta,y;u_\eta)\bigr\}.
\end{split}
    \label{eq:payoff_split_at_exit}
\end{equation}
Since \(J_{\mathcal C}(\tau-\eta,y;u_\eta)
\leq W_{\mathcal C}(\tau-\eta,y)\leq U(\tau-\eta,y)\),
taking the supremum proves the local bound. The minimum of two
upper bounds remains an upper bound and cannot exceed \(U\);
outside \(\mathcal D\), the definition leaves \(U\) unchanged.
At \(\tau=0\), both local branches equal \(g_{\mathcal C}\)
when \eqref{eq:current_upper_boundary_condition} holds.
\end{proof}

The corresponding unavoidable-collision certificates satisfy
\begin{equation}
\begin{split}
    \{U(\tau,\cdot)\leq0\}
    &\subseteq\{(\mathcal P_{\mathcal D}U)(\tau,\cdot)\leq0\}\\
    &\subseteq\mathcal A_{\mathcal C}(\tau).
\end{split}
    \label{eq:local_refinement_avoid_growth}
\end{equation}

\subsection{Local Corrections and Successive Refinement}
\label{subsec:localized_correction}

Define the localized correction by
\begin{equation}
    \Delta_{\mathcal D}[U]:=U-\mathcal P_{\mathcal D}U.
    \label{eq:localized_correction}
\end{equation}
The theorem gives the sharper pointwise statement
\begin{equation}
    0\leq\Delta_{\mathcal D}[U]\leq U-W_{\mathcal C},
    \qquad\Delta_{\mathcal D}[U]=0\text{ outside }\mathcal D.
    \label{eq:localized_correction_nonnegative}
\end{equation}
The correction vanishes at \(\tau=0\) when
\eqref{eq:current_upper_boundary_condition} holds. It removes part
of the current overestimate; it need not recover the entire
common-control or template-extension gap.

Starting from \(U_{\mathcal C}^{(0)}\), a sequence of regions
\(\mathcal D_0,\mathcal D_1,\ldots\) defines
\begin{equation}
\begin{split}
    U_{\mathcal C}^{(k+1)}
      &:=\mathcal P_{\mathcal D_k}U_{\mathcal C}^{(k)}\\
      &=U_{\mathcal C}^{(k)}
          -\Delta_{\mathcal D_k}[U_{\mathcal C}^{(k)}].
\end{split}
    \label{eq:successive_patch_refinement}
\end{equation}

\begin{corollary}[Monotone refinement hierarchy]
\label{cor:monotone_refinement_hierarchy}
For every \(k\geq0\),
\begin{equation}
    W_{\mathcal C}\leq U_{\mathcal C}^{(k+1)}
       \leq U_{\mathcal C}^{(k)}\leq U_{\mathcal C}^{(0)},
    \label{eq:monotone_refinement_hierarchy}
\end{equation}
and the zero-horizon condition is preserved. For each remaining horizon,
\begin{equation}
    \{U_{\mathcal C}^{(0)}\leq0\}
       \subseteq\{U_{\mathcal C}^{(1)}\leq0\}
       \subseteq\cdots\subseteq\{W_{\mathcal C}\leq0\}.
    \label{eq:monotone_certified_avoid_sets}
\end{equation}
\end{corollary}

\begin{proof}
Apply Theorem~\ref{thm:certified_local_refinement} successively.
\end{proof}

After \(K\) exact refinements,
\begin{equation}
    U_{\mathcal C}^{(K)}=U_{\mathcal C}^{(0)}
        -\sum_{k=0}^{K-1}\Delta_{\mathcal D_k}[U_{\mathcal C}^{(k)}].
    \label{eq:modular_exact_representation}
\end{equation}
Regions may overlap, but each correction is defined relative to
the current value, including earlier refinements. Monotonicity alone
does not establish convergence to \(W_{\mathcal C}\) for an arbitrary
choice or sequence of regions.

\subsection{Exactness under Reachable-Tube Containment}
\label{subsec:local_refinement_exactness}

For \(X\subseteq\mathbb R^n\), define
\begin{equation}
    \mathcal R_\tau(X):=
      \{\phi_x^u(t):x\in X,\ u\in\mathbb U_\tau,\ t\in[0,\tau]\}.
    \label{eq:reachable_tube}
\end{equation}

\begin{theorem}[Local exactness]
\label{thm:local_refinement_exactness}
Suppose \(U\) satisfies \eqref{eq:current_global_upper_bound} and
\eqref{eq:current_upper_boundary_condition}. If \(X\subseteq\mathcal D\)
and, for a given \(\tau\in[0,T]\),
\begin{equation}
    \mathcal R_\tau(X)\subseteq\mathcal D,
    \label{eq:reachable_tube_containment}
\end{equation}
then
\begin{equation}
    (\mathcal P_{\mathcal D}U)(\tau,x)=W_{\mathcal C}(\tau,x),
    \qquad x\in X.
    \label{eq:local_refinement_exact_on_X}
\end{equation}
\end{theorem}

\begin{proof}
Containment gives \(\eta_{\mathcal D}^{u}(\tau,x)=\tau\) for every
\(x\in X\) and every admissible control. The continuation branch
is \(e^{-\lambda\tau}g_{\mathcal C}(\phi_x^u(\tau))\), already
included in the stopping minimum. Thus
\(\mathcal L_{\mathcal D}U=W_{\mathcal C}\) on \(X\); taking its
minimum with \(U\geq W_{\mathcal C}\) preserves equality.
\end{proof}

\label{rem:refinement_region_selection}
Containment is sufficient, not necessary. A reachable-tube
over-approximation can provide this condition. Less expensive
selection heuristics include the positive residual
\([r_{\mathcal C}[U]]_+\), overlapping non-neutral singleton values,
disagreement between avoidance controls, and proximity to the zero
level set. These indicators identify candidate regions; they do not
establish tube containment or exactness.

\subsection{Neural Realization and Reuse}
\label{subsubsec:neural_refinement_realization}

Let \(\widehat W_0\) be a learned singleton with exact evaluation
\(\widehat W_0(0,x)=g_0(x)\). Extend it using
\eqref{eq:extended_template_value} and initialize
\begin{equation}
    \widehat U_{\mathcal C}^{(0)}(\tau,x):=
       \min_{1\leq i\leq N}
          \widetilde{\widehat W}_0(\tau,x-Ec_i).
    \label{eq:neural_modular_initialization}
\end{equation}
For each selected region, fix a mask
\(m_k:\mathbb R^n\to[0,1]\) with
\begin{equation}
    m_k(x)=0,\qquad x\notin\mathcal D_k,
    \label{eq:refinement_mask}
\end{equation}
and parameterize a nonnegative correction as
\begin{equation}
    \widehat\Delta_k(\tau,x):=
       \frac{\tau}{T}\,m_k(x)\,
       \rho\bigl(N_{\mathrm{ref},k}(\tau,x;\vartheta_k)\bigr),
    \label{eq:neural_local_correction}
\end{equation}
where \(\rho:\mathbb R\to[0,\infty)\). The raw update is
\begin{equation}
    \widehat V_{\mathcal C}^{(k+1)}
       :=\widehat U_{\mathcal C}^{(k)}-\widehat\Delta_k.
    \label{eq:neural_raw_update}
\end{equation}
Let \(C_\tau\) denote clipping to a fixed numerical interval
\([\underline c_\tau,\overline c_\tau]\):
\(C_\tau(z)=\min\{\overline c_\tau,
\max\{\underline c_\tau,z\}\}\). The evaluated representation is
\begin{equation}
    \widehat U_{\mathcal C}^{(k+1)}(\tau,x):=
    \begin{cases}
       g_{\mathcal C}(x),&\tau=0,\\
       C_\tau\bigl(\widehat V_{\mathcal C}^{(k+1)}(\tau,x)\bigr),
           &\tau>0.
    \end{cases}
    \label{eq:modular_neural_representation}
\end{equation}
Only \(\vartheta_k\) is optimized; the singleton and previous
modules remain frozen. If the current value is already in the same
clipping interval, monotonicity of clipping ensures that the update
is nonincreasing and leaves the exterior of \(\mathcal D_k\)
unchanged. Otherwise exterior invariance holds only for the raw
update. The zero-time branch enforces the boundary value exactly;
continuity as \(\tau\downarrow0\) also depends on the learned
singleton and the clipping rule.

\label{rem:neural_correction_certification}
\label{rem:neural_approximation_scope}
Nonnegativity, localization, and output clipping do not certify
\(W_{\mathcal C}\leq\widehat U_{\mathcal C}^{(k)}\): an excessive
correction can cross below the true value. The raw subtractive update can
reduce local overestimation but cannot repair an existing
underestimate. These properties distinguish the neural architecture
from the exact operator guarantees. Learned accuracy is therefore
evaluated against independent numerical references.

\paragraph{Translated interaction modules}
\label{rem:translated_patch_reuse}
A frozen module can be placed at displacement \(a\in\mathbb R^d\) by
\begin{equation}
    \widehat\Delta_{k,a}(\tau,x)
       :=\widehat\Delta_k(\tau,x-Ea).
    \label{eq:translated_neural_patch}
\end{equation}
For placements \(a_j\), the index \(j\) enumerates locations,
whereas \(k\) identifies the learned module. Each placed copy has support
\(\mathcal D_k+Ea_j\) and shares the same parameters \(\vartheta_k\).
The same shift therefore translates its support and changes no parameters.
Exact equivariance of the underlying correction requires the local
obstacles, refinement region, and continuation data to translate
together. Reuse in different surrounding clutter need not satisfy
that condition, even when a pair has the same relative geometry.
Similarly, summing fixed translated patches is not automatically
equivalent to recomputing each exact correction relative to the
current value. The repeated-pair experiments assess this transfer
empirically.

\subsection{Training Procedure}
\label{subsec:algorithmic_realization}

Local neural training approximates a Bellman correction; it does
not solve the first-exit problem \eqref{eq:local_continuation_operator}
exactly. Starting states are sampled in the selected region, while
short trajectories and continuation queries may leave it. Targets
use the clutter cost and the current assembled value, with the
existing modules frozen and the target evaluation detached from the
online gradient. Section~\ref{subsec:experimental_setup} specifies
the reported action discretization, rollout, and training settings.
A rollout that leaves the patch within a training step is not
claimed to implement exact first-exit stopping.

To distinguish the sampled update from the exact operator, let
\(\Phi_h(x,a)\) be the numerical flow for a held action
\(a\in\mathcal U_{\mathrm d}\), and let
\(\widehat S_{\mathcal C}(h,x;a)\) denote the corresponding numerical
approximation of the discounted segment minimum. For a detached target
representation \(\Psi^-\), the sampled Bellman target has the form
\begin{equation}
\begin{split}
 y(\tau,x)=\max_{a\in\mathcal U_{\mathrm d}}\min\bigl\{
   &\widehat S_{\mathcal C}(h,x;a),\\
   &e^{-\lambda h}\Psi^-(\tau-h,\Phi_h(x,a))\bigr\},
\end{split}
\label{eq:sampled_bellman_target}
\end{equation}
where \(h=\min\{0.1,\tau\}\) in the numerical study and
\(y(0,x)=g_{\mathcal C}(x)\). For singleton training, replace
\(g_{\mathcal C}\) by \(g_0\). During patch training, \(\Psi^-\)
contains the frozen singleton composition and the target residual;
only the online residual parameters are updated. Thus the learned object
is the assembled clutter value, not a residual fitted to grid-reference
labels. The finite action set, numerical propagation, and segment-minimum
approximation distinguish this target from \(\mathcal T_{\mathcal C}\).

\begin{figure}[t]
\centering
\begin{minipage}{0.96\columnwidth}
\hrule\smallskip
\refstepcounter{algorithm}\label{alg:modular_refinement}
\textbf{Algorithm \thealgorithm: Modular neural refinement}
\smallskip\hrule\smallskip
\textbf{Input:} extended singleton \(\widetilde{\widehat W}_0\),
clutter \(\mathcal C\), refinement budget \(K\).
\begin{enumerate}
  \item Set \(k\gets0\), freeze the singleton, and initialize
        \(\widehat U\) by
        \eqref{eq:neural_modular_initialization}.
  \item Select a region \(\mathcal D_k\) and mask \(m_k\); stop
        if no region is selected or the budget is exhausted.
  \item Freeze all existing modules and train only
        \(\vartheta_k\) using local Bellman targets for the
        assembled value.
  \item Apply \eqref{eq:neural_local_correction}--\eqref{eq:modular_neural_representation},
        freeze the new module, set \(k\gets k+1\), and repeat
        from step 2.
\end{enumerate}
\textbf{Output:} the assembled value \(\widehat U\).
\smallskip\hrule
\end{minipage}
\end{figure}

The experiments instantiate one prescribed patch and one singleton,
then reuse this fixed library by translation. They do not evaluate
automatic region selection or arbitrary sequences of newly trained
modules.



\section{Experiments}
\label{sec:experiments}

The evaluation follows the structure of the representation: validate
one singleton, isolate the discrepancy introduced by clutter, and reuse
the frozen singleton and local correction in larger translated scenes.
Throughout, \emph{composition} denotes the minimum of translated
singleton predictions; \emph{modular} adds the local interaction
module; and a \emph{global NN} represents a complete scene with one
network. A separate PiCar-X demonstration evaluates online singleton
reuse in physical closed-loop operation.

\subsection{Experimental Setup}
\label{subsec:experimental_setup}

We use the Dubins system with state \(x=(p_x,p_y,\theta)\),
\begin{equation}
    \dot p_x=v\cos\theta,\qquad
    \dot p_y=v\sin\theta,\qquad
    \dot\theta=u,
    \label{eq:exp_dubins}
\end{equation}
where \(v=1~\mathrm{m/s}\), \(u\in[-1,1]~\mathrm{rad/s}\),
\(T=1~\mathrm{s}\), and \(\lambda=0.2\).
Heading is represented periodically on \([-\pi,\pi]\).
Circular obstacles have radius \(0.5~\mathrm{m}\). The saturated
signed-distance reach cost takes values in \([-1,1]\), with inner
and outer ramp lengths \(0.5~\mathrm{m}\) and \(1~\mathrm{m}\),
respectively; clutter uses \(g_{\mathcal C}=\min_i g_i\).
All numerical classifications use the nonpositive level set of the
corresponding value approximation.

\paragraph{Neural learning}
All networks are SIRENs with two hidden layers of width 128,
\(w_0=30\), and sine/cosine heading encoding. Each training run
uses 10,000 updates. The Bellman target uses
\(h=\min\{0.1,\tau\}\), controls
\(\{-1,-0.5,0,0.5,1\}\), and RK4 propagation with eight substeps;
continuation is evaluated at \(\tau-h\). Targets are detached from
the online-network gradient. Soft target averaging has mixing
coefficient \(0.02\); the separate soft reachability-constraint
penalty has weight \(0.1\). Table~\ref{tab:training_settings}
gives the remaining model-specific settings. The singleton and patch
are each one seed-7 realization. Seeds \(7,19,42\) vary only global
training in the scaling study; the library and layouts remain fixed.

\begin{table}[t]
    \centering
    \caption{Training settings. All networks have two width-128 hidden
    layers and use 10,000 updates per training run.}
    \label{tab:training_settings}
    \small
    \begin{tabular}{@{}lcc@{}}
        \toprule
        Model & Batch size & Learning rate \\
        \midrule
        Singleton & 512 & \(3\times10^{-4}\) \\
        Residual patch & 512 & \(3\times10^{-4}\) \\
        Global NN & 512 & \(3\times10^{-4}\) \\
        \bottomrule
    \end{tabular}
\end{table}

The singleton and global models use replay capacity \(10^5\),
with collection sizes of 1024. The residual is
supported on
\begin{equation}
    p_x\in[-0.5,0.5],\qquad p_y\in[-2,0]~\mathrm{m}.
    \label{eq:patch_support}
\end{equation}
It uses the nonnegative subtractive architecture of
Section~\ref{subsubsec:neural_refinement_realization}, including
\(\tau/T\) scaling, output clipping, and exact zero-time evaluation.
Its spatial mask vanishes on and outside the patch boundary and has
feather fraction \(0.25\). Only the residual parameters are trained.
Starting states lie inside the patch and remaining times are uniform
on \([0,T)\), without a time curriculum; Bellman rollouts may leave
the patch.


\paragraph{Numerical references}
Independent grid calculations are used only for evaluation, never as
training labels. They use 128 headings, \(\Delta t=0.025~\mathrm{s}\),
17 control samples, 16 stopping substeps, exact constant-turn arcs,
and periodic trilinear interpolation. The singleton reference has a
\(161\times161\) spatial grid and one-metre padding. The two-obstacle
reference has 241 spatial points per axis on its padded domain. Scaling
references use spatial spacing \(\leq0.05~\mathrm{m}\) with the same
heading, time, control, and stopping resolutions.

\paragraph{Metrics}
RMSE compares predicted values with the fine-grid reference;
\emph{near-zero RMSE} uses only samples with
\(\lvert W_{\mathrm{grid}}\rvert\leq0.15\).
Sign disagreement counts differing avoid classifications.
The \emph{missed-avoid rate} is the fraction of reference-avoid
samples classified as safe by the prediction. Thus its denominator
is the number of reference-avoid samples, not the entire population.
Aggregates exclude \(\tau=0\) unless stated otherwise.

\subsection{Experiment 1: Singleton Validation}
\label{subsec:exp_singleton}

The singleton is trained for an obstacle at the origin on
\((p_x,p_y)\in[-2.5,2.5]^2\) and \(\theta\in[-\pi,\pi]\).
We compare it with the grid reference at \(\theta=0\), on
\(161\times161\) positions and
\(\tau\in\{0,0.25,0.5,1\}~\mathrm{s}\).
Table~\ref{tab:singleton_validation} reports the nonzero-horizon
errors; Fig.~\ref{fig:singleton_validation} compares all four
remaining-horizon slices. This singleton is then frozen for all numerical composition
and refinement experiments.

\begin{table}[t]
    \centering
    \caption{Singleton validation at \(\theta=0\).
    The aggregate uses the three nonzero horizons.}
    \label{tab:singleton_validation}
    \small
    \setlength{\tabcolsep}{4pt}
    \begin{tabular}{@{}cccc@{}}
        \toprule
        \(\tau\) [s] & RMSE & Max.\ abs. & Sign error [\%] \\
        \midrule
        0.25 & 0.01679 & 0.20485 & 0.120 \\
        0.50 & 0.01826 & 0.19050 & 0.150 \\
        1.00 & 0.04226 & 0.17568 & 0.235 \\
        \midrule
        Aggregate & 0.02829 & 0.20485 & 0.168 \\
        \bottomrule
    \end{tabular}
\end{table}

\begin{figure*}[!t]
    \centering
    \includegraphics[width=\textwidth]{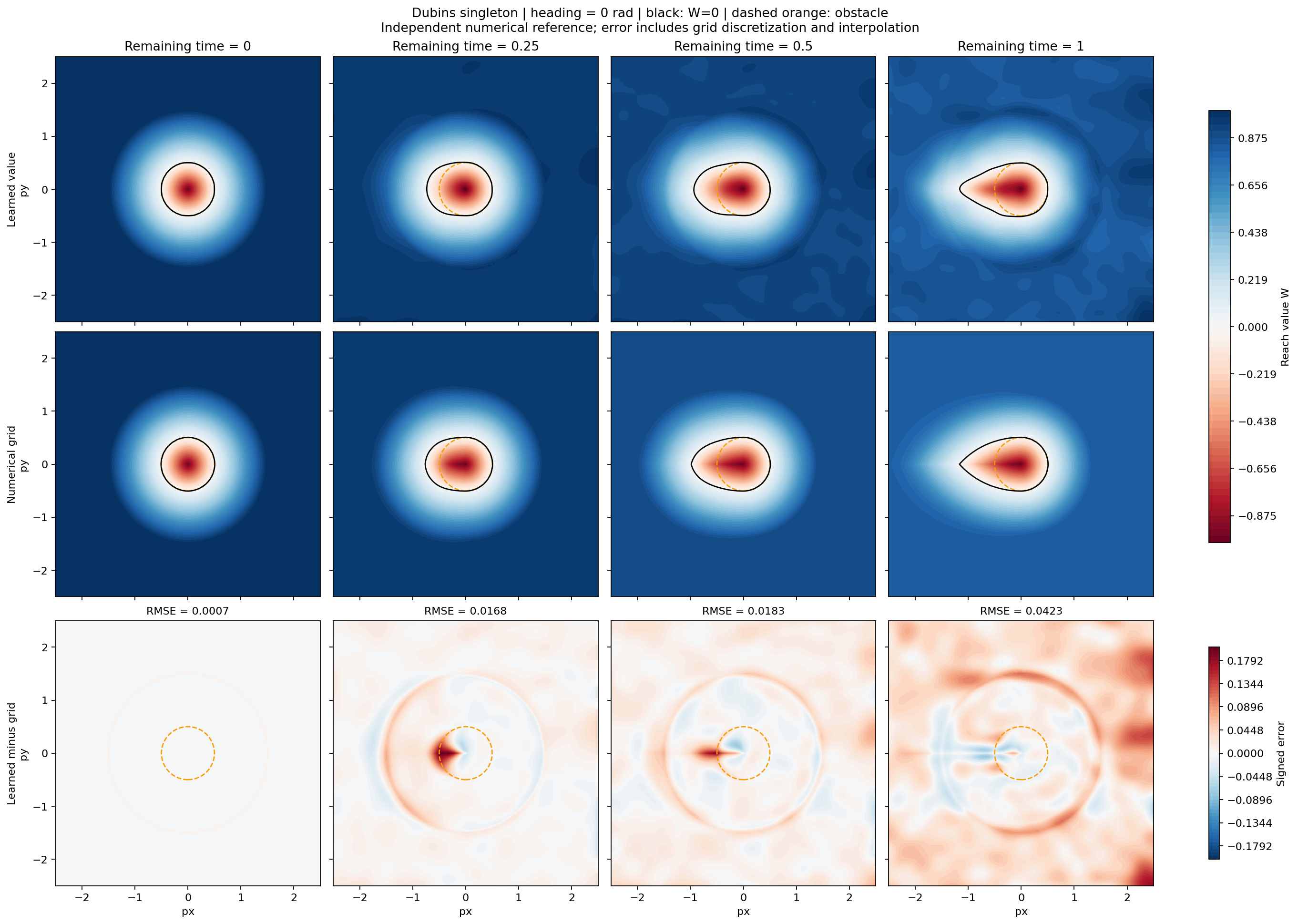}
    \caption{Singleton validation at heading \(\theta=0\).
    Columns correspond to remaining horizons
    \(\tau\in\{0,0.25,0.5,1\}~\mathrm{s}\).
    Rows show the learned value, independent numerical reference, and
    signed prediction error. Black curves mark zero level sets and
    dashed orange curves denote the obstacle. The zero-time condition is
    imposed exactly; small discrepancies against the interpolated grid
    at that slice are reference-comparison effects.
    Table~\ref{tab:singleton_validation} quantifies the nonzero horizons.}
    \label{fig:singleton_validation}
\end{figure*}


\subsection{Experiment 2: Composition and Local Refinement}
\label{subsec:exp_phase2}

To separate singleton approximation error from the effect of sharing
one control across obstacles, we consider two centers
\(c_1=(-a,0)\), \(c_2=(a,0)\), with
\(a=1.75,1.0,0.4~\mathrm{m}\). These give the far, cost-band-overlap,
and physical-overlap layouts, respectively.
The \emph{composed NN} takes the minimum of translated singleton
predictions; the \emph{composed grid} takes the minimum of translated
singleton numerical values. Both are compared with the independent
full-clutter grid. The latter comparison isolates a discrepancy that
does not rely on a singleton NN, subject to numerical-reference error.

Table~\ref{tab:phase2_composition} aggregates \(81\times81\)
positions, eight headings, and
\(\tau\in\{0.25,0.5,1\}~\mathrm{s}\): 157,464 state-time
values per layout. For the physical-overlap layout, one residual is
trained on \eqref{eq:patch_support} while the singleton remains frozen.
Table~\ref{tab:patch_ablation} compares this same composition before
and after refinement, using the regular eight-heading population,
its subset inside the patch, and a fresh holdout of 15,000 state-time
evaluations from 5,000 off-grid states. The singleton and residual then
form the frozen library used in Experiment~3.

\begin{table*}[!t]
    \centering
    \caption{Composition error against the full-clutter numerical
    reference across eight headings and three nonzero horizons.
    The composed NN uses the frozen seed-7 singleton.}
    \label{tab:phase2_composition}
    \small
    \begin{tabular}{llccc}
        \toprule
        Layout & Method & RMSE & Near-zero RMSE & Missed avoid [\%] \\
        \midrule
        Far & Composed NN & 0.02013 & 0.02109 & 2.795 \\
            & Composed grid & 0.00004 & 0.000001 & 0.000 \\
        \midrule
        Cost-band overlap & Composed NN & 0.02370 & 0.02159 & 2.863 \\
            & Composed grid & 0.01090 & 0.00020 & 0.000 \\
        \midrule
        Physical overlap & Composed NN & 0.02469 & 0.03949 & 2.992 \\
            & Composed grid & 0.01300 & 0.03447 & 1.520 \\
        \bottomrule
    \end{tabular}
\end{table*}

\begin{table}[t]
    \centering
    \caption{Local refinement in the physical-overlap layout.
    ``All'' comprises all eight evaluated headings and three horizons;
    ``inside'' restricts to the patch; ``fresh'' is the off-grid holdout.
    NZ is near-zero RMSE; Miss is missed avoid.}
    \label{tab:patch_ablation}
    \small
    \setlength{\tabcolsep}{3pt}
    \begin{tabular}{@{}llccc@{}}
        \toprule
        Population & Method & RMSE & NZ & Miss [\%] \\
        \midrule
        All & Comp. & 0.02469 & 0.03949 & 2.992 \\
            & + patch & 0.02277 & 0.03259 & 2.541 \\
        \midrule
        Inside & Comp. & 0.04249 & 0.05843 & 4.114 \\
            & + patch & 0.02280 & 0.02426 & 2.296 \\
        \midrule
        Fresh & Comp. & 0.02395 & 0.03489 & 3.021 \\
            & + patch & 0.02201 & 0.02563 & 2.789 \\
        \bottomrule
    \end{tabular}
\end{table}

\begin{figure*}[!t]
    \centering
    \includegraphics[width=\textwidth]{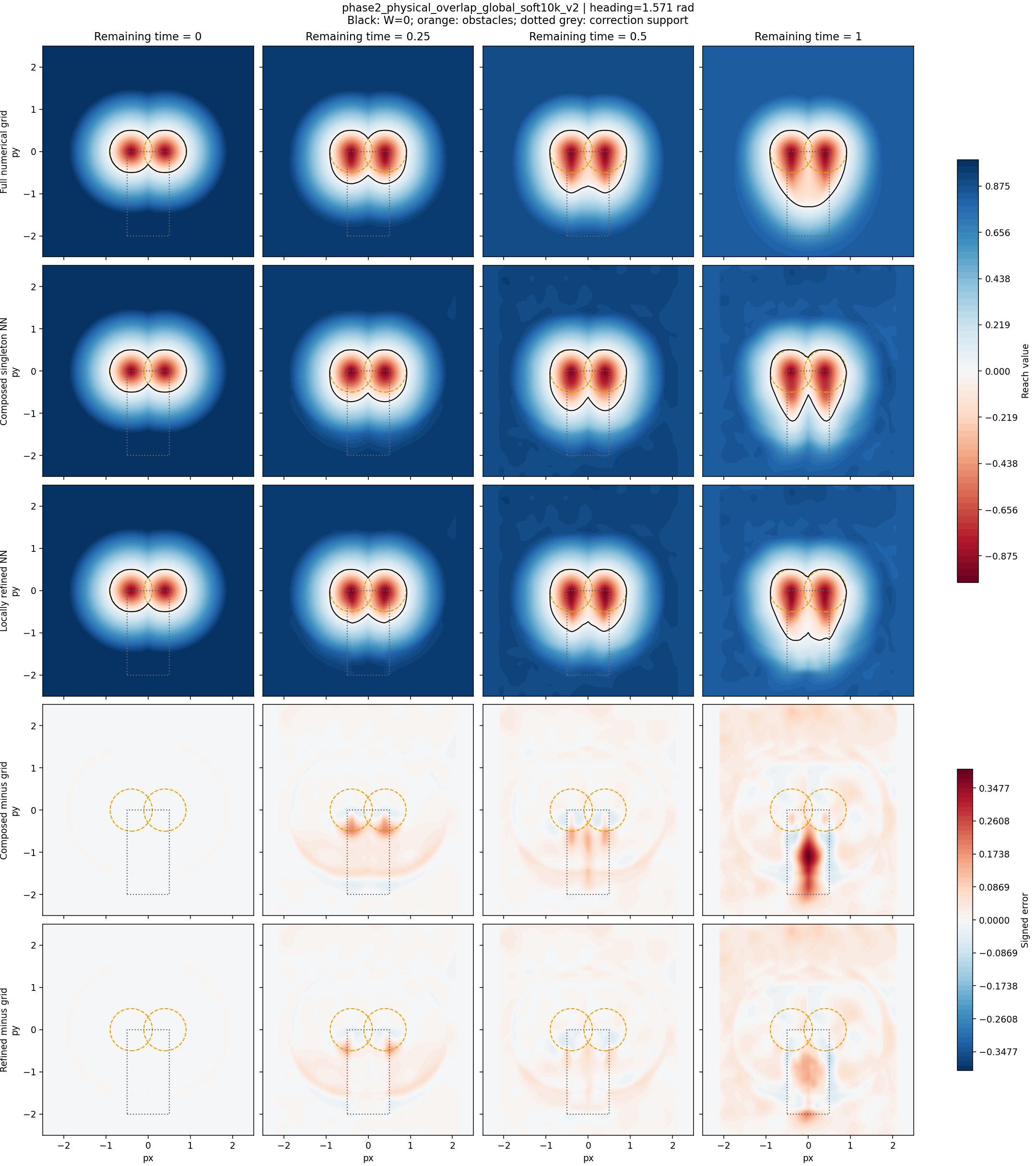}
    \caption{Local interaction refinement for the physical-overlap
    configuration at \(\theta=\pi/2\). Columns show
    \(\tau\in\{0,0.25,0.5,1\}~\mathrm{s}\). The upper three rows
    compare the full-clutter numerical reference, singleton composition,
    and composition with the learned patch. The last two rows show
    signed errors before and after refinement. The dotted grey rectangle
    marks \(p_x\in[-0.5,0.5]\), \(p_y\in[-2,0]~\mathrm{m}\).
    The largest improvements occur inside this region; the maximum
    recorded exterior change is zero. Table~\ref{tab:patch_ablation}
    reports the eight-heading and independent holdout comparisons.}
    \label{fig:phase2_patch}
\end{figure*}


\subsection{Experiment 3: Scaling and Translation Reuse}
\label{subsec:exp_phase3}

Square workspaces have areas \(300,600,1000~\mathrm{m}^2\), denoted
\(3\times,6\times,10\times\) relative to \(100~\mathrm{m}^2\).
The \emph{fixed-count} arm keeps five obstacles; the
\emph{fixed-density} arm uses 15, 30, and 50 obstacles, respectively.
Within each area and arm, we evaluate initial and relocated scenes.
Repeated interacting pairs retain the trained \(0.8~\mathrm{m}\)
separation and orientation. Thus the study tests reuse under
translation of repeated geometry, with unchanged horizon and heading
range.

For each initial scene, a global NN is trained for 10,000 updates.
After relocation, \emph{scratch} training starts a new global network;
\emph{warm} training copies the corresponding initial model, resets
the optimizer and replay buffer, and synchronizes the target network.
Both receive 10,000 adaptation updates. Global training uses three
seeds; the modular method repositions the same frozen library with
zero additional training.

\emph{Uniform} evaluation uses 4096 states at three nonzero horizons
(12,288 state-time values). \emph{Obstacle-local} evaluation uses 1024
matching offsets in radius-\(2.5~\mathrm{m}\) disks per obstacle;
\emph{interaction} evaluation uses 1024 matching offsets per designated
pair in \(3\times4~\mathrm{m}\) rectangles. Samples outside the
workspace are removed; overlapping local populations retain their
mixture contributions. Methods and seeds share the same probes within
a scene. Local populations prevent additional empty space from
concealing obstacle and interaction errors. Adaptation curves use a
fixed changed-region probe population for each scene.


\subsection{Hardware Feasibility Demonstration}
\label{subsec:hardware_experiment}

We evaluate online singleton reuse on a SunFounder PiCar-X driving
among six traffic cones (Fig.~\ref{fig:picarx_hardware}). OptiTrack
provides vehicle and cone poses through UnifiedMocapRouter to a Python
controller running under WSL on a Windows laptop. The controller
transmits steering and motor commands over the local network.
The physical experiment uses the same Dubins reach-cost formulation,
Bellman-training procedure, and translation-based singleton composition
as the numerical studies, but a \emph{separately trained,
hardware-parameterized singleton network}; the numerical residual
patch is not deployed.

\paragraph{Hardware-specific reachability model}
The model has the Dubins dynamics in~\eqref{eq:exp_dubins}, with
$v=0.18~\mathrm{m/s}$, $|u|\leq0.35~\mathrm{rad/s}$,
$T_{\mathrm{hw}}=3~\mathrm{s}$, and
$\lambda=0.2~\mathrm{s}^{-1}$. Each cone is represented in the
reach cost by a circular collision region with \emph{inflated radius}
$r_{\mathrm{hw}}=0.20~\mathrm{m}$, accounting for the physical cone
radius and vehicle clearance. This is the radius used to learn the
hardware singleton avoid value, rather than merely a plotting margin.
The training domain is $(p_x,p_y)\in[-2,2]^2~\mathrm{m}^2$,
$\psi\in[-\pi,\pi]$, with remaining horizon $\tau\in[0,3]~\mathrm{s}$.
The five network inputs are normalized time and position, followed by
$\sin\psi$ and $\cos\psi$.

The hardware singleton is a SIREN with two width-128 hidden layers and
sine frequency $w_0=30$. The same reach-cost Bellman-learning method
as in Section~\ref{subsec:experimental_setup} is used, with hardware
control samples
$\mathcal U_{\mathrm f}=\{0,\pm0.1166667,\pm0.2333333,\pm0.35\}$
$\mathrm{rad/s}$, a $0.1~\mathrm{s}$ Bellman step, and RK4 integration
with eight substeps. Training uses Adam with constant learning rate
$3\times10^{-4}$, 10,000 iterations, batch size 1024, and seed 7.
The trained network is frozen and evaluated at obstacle-relative states
for all six cones. Thus the hardware uses the \emph{same learning and
reuse mechanism}, not the identical network weights or physical
parameters of the numerical studies.

\paragraph{Waypoint tracking and safety filter}
A nominal PID heading controller follows four cyclic waypoints with
gains $K_P=1.5$, $K_I=0$, and $K_D=0.05$. The motor-command setting
is 8. At each update the controller computes the composed hardware
singleton value at the current state and at the state predicted under
the PID command. Write $\widehat U_{\mathcal C,\mathrm{hw}}^{(0)}$
for the minimum of the six translated hardware singleton predictions,
$F_h(x,u)$ for the Dubins endpoint prediction, and
$h_{\mathrm f}=0.1~\mathrm{s}$ for the one-step prediction interval.
Define
\begin{equation}
\begin{aligned}
S_{\mathrm{now},t}
 &=\widehat U_{\mathcal C,\mathrm{hw}}^{(0)}
   (T_{\mathrm{hw}},x_t),\\
Q_t(u)
 &=\widehat U_{\mathcal C,\mathrm{hw}}^{(0)}
   (T_{\mathrm{hw}},F_{h_{\mathrm f}}(x_t,u)),\\
S_{\mathrm{PID},t}&=Q_t(u_{\mathrm{PID},t}).
\end{aligned}
\label{eq:hardware_filter}
\end{equation}
The PID command is issued if both the current and predicted values
are at least $\varepsilon=0.02$. Otherwise, the filter intervenes
with the greedy action among the same seven yaw-rate candidates used
for hardware-network training:
\begin{equation}
\begin{aligned}
S_{\min,t}&:=\min\{S_{\mathrm{now},t},S_{\mathrm{PID},t}\},\\
u_{\mathrm{greedy},t}&\in\arg\max_{u\in\mathcal U_{\mathrm f}}Q_t(u),\\
u_t&=\begin{cases}
u_{\mathrm{PID},t},&S_{\min,t}\geq\varepsilon,\\
u_{\mathrm{greedy},t},&S_{\min,t}<\varepsilon.
\end{cases}
\end{aligned}
\label{eq:hardware_switch}
\end{equation}
Both current- and successor-state evaluations query the fixed learned
$3~\mathrm{s}$ horizon; the $0.1~\mathrm{s}$ prediction interval is
distinct from that reachability horizon. This is an online greedy
intervention rule, not an exact Bellman backup at remaining time
$T_{\mathrm{hw}}-h_{\mathrm f}$. The threshold and greedy choice do not
ensure that an admissible action with positive value exists, nor do
they provide a formal closed-loop safety guarantee.

\begin{figure*}[!t]
    \centering
    \includegraphics[width=\textwidth]{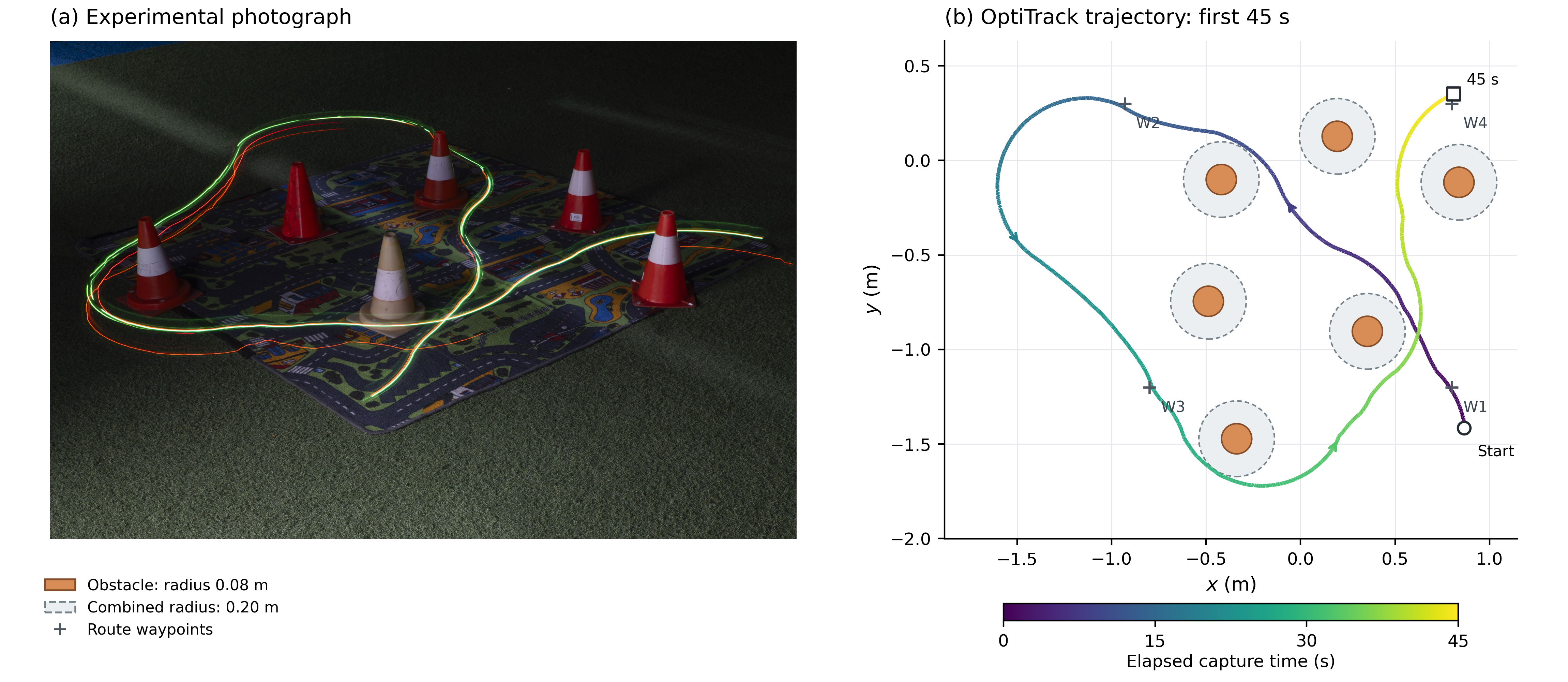}
    \caption{Physical PiCar-X demonstration among six traffic cones.
    (a) Experimental photograph, not synchronized with the plotted
    trajectory. (b) First $45~\mathrm{s}$ of unsmoothed OptiTrack
    positions, colored by capture time; crosses indicate waypoints.
    Solid $0.08~\mathrm{m}$ circles show physical cone footprints,
    while dashed $0.20~\mathrm{m}$ circles indicate the inflated
    collision radius used in the hardware reach-cost model and
    singleton avoid-value training. The same frozen singleton network
    is evaluated relative to all six cones, with nominal PID tracking
    and greedy reachability-based intervention.}
    \label{fig:picarx_hardware}
\end{figure*}


\section{Results and Discussion}
\label{sec:results}

\subsection{Singleton Approximation}
\label{subsec:results_singleton}

Across the three nonzero horizons at \(\theta=0\), the singleton
has RMSE 0.02829 and sign disagreement 0.168\%
(Table~\ref{tab:singleton_validation}). RMSE increases from 0.01679
at \(0.25~\mathrm{s}\) to 0.04226 at \(1~\mathrm{s}\);
this is consistent with longer Bellman propagation but does not
isolate its cause. At zero time the representation returns the exact
reach cost; its RMSE of \(7.30\times10^{-4}\) against the interpolated
grid reflects the numerical-reference comparison.

Independent diagnostic probes give held-out Bellman RMSE 0.00911,
but a maximum value above the reach cost of 0.02884. Consequently,
enforcing the zero-horizon condition and penalizing constraint
violations do not provide a one-sided neural value certificate.

\subsection{Composition Error and Interaction-Dependent Refinement}
\label{subsec:results_phase2}

Table~\ref{tab:phase2_composition} distinguishes value error from
avoid-set disagreement. In the far layout, composed-grid and
full-clutter values nearly coincide (RMSE \(3.54\times10^{-5}\)),
with no missed-avoid samples. Cost-band overlap increases RMSE to
0.01090 while near-zero RMSE remains \(2.03\times10^{-4}\), again
without missed-avoid samples. Physical overlap instead gives
near-zero RMSE 0.03447 and 1.52\% missed avoid. The discrepancy
persists without a singleton NN, identifying the interaction regime
addressed by refinement.

The strongest improvement occurs where the correction acts:
inside the patch, RMSE decreases from 0.04249 to 0.02280
(46.3\%) and near-zero RMSE from 0.05843 to 0.02426 (58.5\%).
Across the full eight-heading population, the corresponding reductions
are 7.8\% and 17.5\%. The fresh off-grid holdout also improves,
including a near-zero RMSE reduction from 0.03489 to 0.02563
(Table~\ref{tab:patch_ablation}). Missed-avoid rates decrease in all
three reported populations. The maximum recorded change outside the
patch is zero.

These results support the intended division of information: singleton
composition carries the reusable structure, while the local residual
corrects a discrepancy concentrated in the interacting region
(Fig.~\ref{fig:phase2_patch}). The patch therefore need not produce a
large whole-domain RMSE change to represent useful information.
Residual approximation errors remain.

\subsection{Scaling and Translation Reuse}
\label{subsec:results_scaling}

Table~\ref{tab:phase3_local_scaling} and
Fig.~\ref{fig:phase3_scaling} evaluate the frozen library after
relocation. Global entries are means and sample standard deviations
over three training seeds; modular entries use the same fixed library.
Across all six area/arm combinations, modular interaction RMSE stays
between 0.02597 and 0.02639, and obstacle-local RMSE between 0.01797
and 0.01864. Mean scratch interaction RMSE increases from 0.07010
to 0.08739 in the fixed-count arm and from 0.07181 to 0.09756 in
the fixed-density arm. Thus the learned components retain similar
local accuracy when translated into the tested larger scenes, without
relearning their parameters.

The fixed-count arm separates spatial extent from obstacle count;
the fixed-density arm increases both. Figure~\ref{fig:phase3_scaling}
also reports uniform-domain error. With five obstacles, a larger
workspace contains more empty space, so uniform RMSE alone can
obscure local inaccuracies. The local metrics therefore measure
retention of the reusable geometry rather than a benefit caused by
diluting the evaluation population.

\begin{table*}[!t]
    \centering
    \caption{Relocated-scene local accuracy versus area.
    Global entries are mean \(\pm\) sample SD over three training seeds;
    the modular library is one fixed realization. ``Obs.'' is
    obstacle-local RMSE, ``Int.'' is interaction RMSE, ``NZ'' is
    interaction near-zero RMSE, and ``Miss'' is the interaction
    missed-avoid rate. Area multipliers are relative to
    \(100~\mathrm{m}^2\).}
    \label{tab:phase3_local_scaling}
    \footnotesize
    \setlength{\tabcolsep}{4pt}
    \begin{tabular}{cclcccc}
        \toprule
        Arm & Area & Method & Obs.\ RMSE & Int.\ RMSE
        & Int.\ NZ & Miss [\%] \\
        \midrule
        Fixed count
        & \(3\times\) & Modular
        & 0.01815 & 0.02633 & 0.03874 & 2.42 \\
        & & Scratch
        & \(0.05532\pm0.00165\)
        & \(0.07010\pm0.00165\)
        & \(0.06837\pm0.00441\)
        & \(4.88\pm0.37\) \\
        & & Warm
        & \(0.05538\pm0.00243\)
        & \(0.07176\pm0.00326\)
        & \(0.06927\pm0.00734\)
        & \(5.04\pm0.35\) \\
        \cmidrule{2-7}
        & \(6\times\) & Modular
        & 0.01797 & 0.02597 & 0.03807 & 2.13 \\
        & & Scratch
        & \(0.06146\pm0.00209\)
        & \(0.07923\pm0.00247\)
        & \(0.07758\pm0.00186\)
        & \(7.07\pm0.83\) \\
        & & Warm
        & \(0.06210\pm0.00083\)
        & \(0.08174\pm0.00068\)
        & \(0.08192\pm0.00074\)
        & \(7.59\pm0.24\) \\
        \cmidrule{2-7}
        & \(10\times\) & Modular
        & 0.01817 & 0.02636 & 0.03880 & 2.42 \\
        & & Scratch
        & \(0.06800\pm0.00072\)
        & \(0.08739\pm0.00138\)
        & \(0.09093\pm0.00344\)
        & \(9.27\pm0.66\) \\
        & & Warm
        & \(0.06850\pm0.00180\)
        & \(0.09009\pm0.00354\)
        & \(0.09944\pm0.00567\)
        & \(10.47\pm0.34\) \\
        \midrule
        Fixed density
        & \(3\times\) & Modular
        & 0.01844 & 0.02603 & 0.03810 & 2.35 \\
        & & Scratch
        & \(0.05968\pm0.00184\)
        & \(0.07181\pm0.00116\)
        & \(0.07102\pm0.00129\)
        & \(4.41\pm0.55\) \\
        & & Warm
        & \(0.05779\pm0.00112\)
        & \(0.07058\pm0.00052\)
        & \(0.07019\pm0.00116\)
        & \(4.21\pm0.10\) \\
        \cmidrule{2-7}
        & \(6\times\) & Modular
        & 0.01839 & 0.02600 & 0.03807 & 2.18 \\
        & & Scratch
        & \(0.08021\pm0.00570\)
        & \(0.08758\pm0.00404\)
        & \(0.08271\pm0.00410\)
        & \(7.77\pm0.93\) \\
        & & Warm
        & \(0.08162\pm0.00325\)
        & \(0.08877\pm0.00169\)
        & \(0.08160\pm0.00299\)
        & \(7.55\pm0.29\) \\
        \cmidrule{2-7}
        & \(10\times\) & Modular
        & 0.01864 & 0.02639 & 0.03878 & 2.46 \\
        & & Scratch
        & \(0.08957\pm0.00350\)
        & \(0.09756\pm0.00348\)
        & \(0.09952\pm0.00455\)
        & \(10.34\pm0.66\) \\
        & & Warm
        & \(0.08891\pm0.00075\)
        & \(0.09623\pm0.00066\)
        & \(0.09656\pm0.00105\)
        & \(9.71\pm0.73\) \\
        \bottomrule
    \end{tabular}
\end{table*}

\begin{figure*}[!t]
    \centering
    \includegraphics[width=\textwidth]{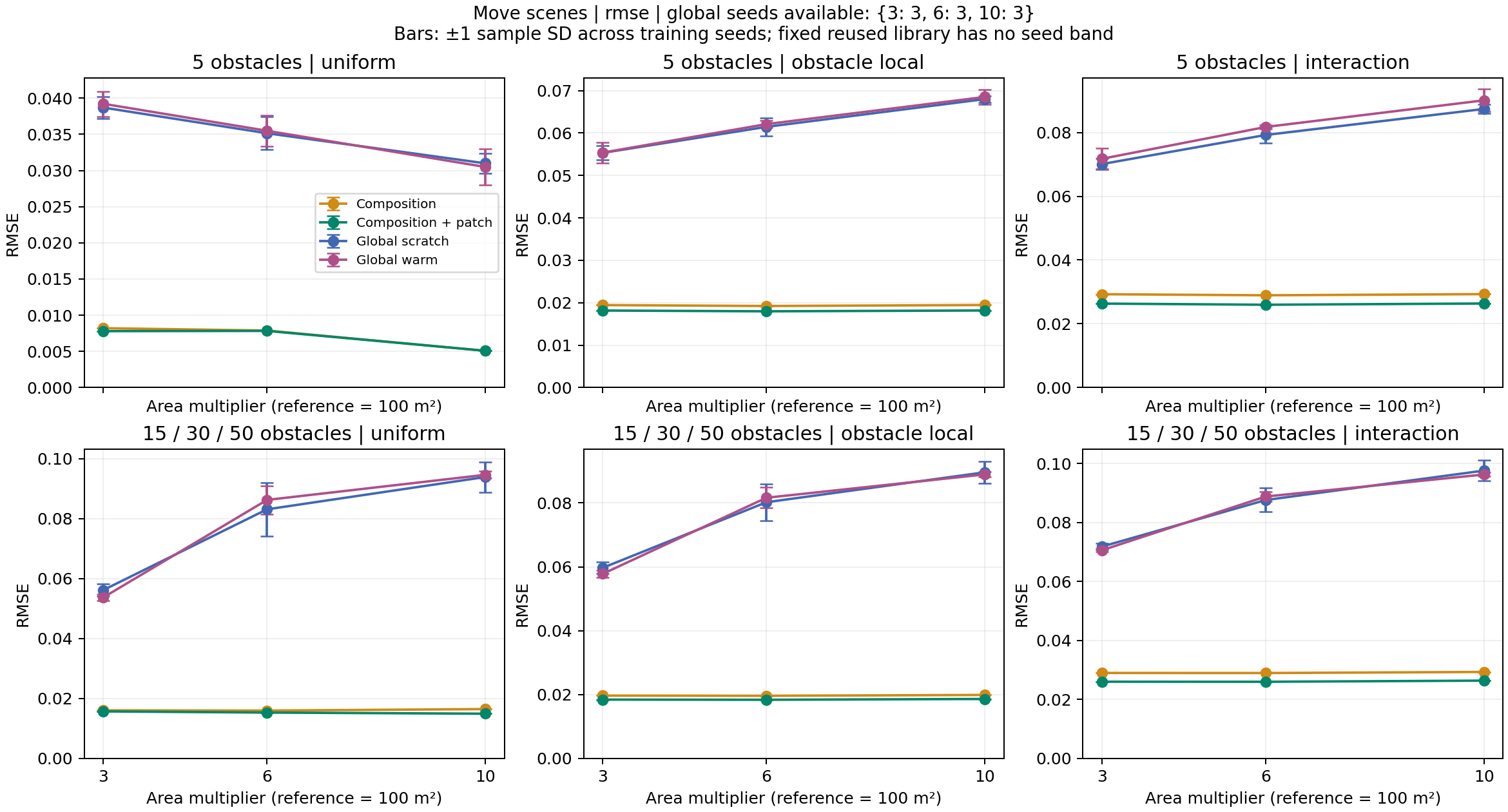}
    \caption{Relocated-scene RMSE versus workspace area.
    The top row keeps five obstacles; the bottom row maintains density
    using 15, 30, and 50 obstacles. Columns report uniform-domain,
    obstacle-local, and interaction RMSE. Areas \(3,6,10\) correspond
    to \(300,600,1000~\mathrm{m}^2\). Global curves show means and
    one sample standard deviation over seeds \(7,19,42\); both reusable
    predictions use one frozen library. Global models receive 10,000
    relocation updates; the library receives none. Local populations
    prevent increasing empty area from concealing interaction errors.}
    \label{fig:phase3_scaling}
\end{figure*}

\subsection{Interaction Correction Under Reuse}
\label{subsec:results_phase3_patch}

Table~\ref{tab:phase3_patch_ablation} reports the direct comparison
between composition and its translated local correction. On the same
relocated interaction populations, the translated patch reduces RMSE by approximately
10\%, near-zero RMSE by 7.2--7.4\%, and missed-avoid rate by
22--24\%. In the 50-obstacle, \(10\times\) fixed-density scene,
interaction RMSE decreases from 0.02933 to 0.02639, near-zero RMSE
from 0.04181 to 0.03878, and missed avoid from 3.16\% to 2.46\%.
The interaction-specific correction therefore remains useful when the
same pair geometry reappears elsewhere.

Whole-domain averages can weight this correction weakly. In the
\(10\times\) fixed-count scene, uniform RMSE changes from 0.005071
to 0.005082, while the corresponding interaction RMSE improves from
0.02930 to 0.02636. Added empty space and the local support explain
why uniform accuracy alone cannot assess the refinement's intended
role; the observed local benefit does not imply improvement of every
aggregate metric.

\begin{table}[t]
\centering
\caption{Translated-patch ablation on relocated interaction populations.
All entries use the same frozen library. Each arrow denotes
composition \(\rightarrow\) composition plus patch. NZ is near-zero
RMSE; Miss is the missed-avoid rate.}
\label{tab:phase3_patch_ablation}
\footnotesize
\setlength{\tabcolsep}{2.4pt}
\begin{tabular}{@{}ccccc@{}}
\toprule
Arm & Area & RMSE & NZ & Miss [\%] \\
\midrule
Count & \(3\times\) & \(.02927\to.02633\) & \(.04178\to.03874\) & \(3.10\to2.42\) \\
Count & \(6\times\) & \(.02893\to.02597\) & \(.04112\to.03807\) & \(2.81\to2.13\) \\
Count & \(10\times\)& \(.02930\to.02636\) & \(.04183\to.03880\) & \(3.10\to2.42\) \\
\midrule
Density & \(3\times\) & \(.02899\to.02603\) & \(.04114\to.03810\) & \(3.03\to2.35\) \\
Density & \(6\times\) & \(.02897\to.02600\) & \(.04112\to.03807\) & \(2.86\to2.18\) \\
Density & \(10\times\)& \(.02933\to.02639\) & \(.04181\to.03878\) & \(3.16\to2.46\) \\
\bottomrule
\end{tabular}
\end{table}


\subsection{Modular Reuse versus Global Warm-Start Adaptation}
\label{subsec:results_adaptation}

\paragraph{Retained local accuracy after relocation}
Across all six area/arm combinations in
Table~\ref{tab:phase3_local_scaling}, modular reuse has lower
obstacle-local RMSE, interaction RMSE, interaction near-zero RMSE,
and missed-avoid rate than the mean global warm-start baseline.
For interaction RMSE, the relative reduction
\(100(1-E_{\mathrm{mod}}/\overline E_{\mathrm{warm}})\)
is 63.3\%, 68.2\%, and 70.7\% in the fixed-count arm, and
63.1\%, 70.7\%, and 72.6\% in the fixed-density arm.
These percentages are calculated from the reported rounded table
entries; they compare one fixed library with the three-seed
warm-start means, not replicated library-training outcomes.
Crucially, the modular library receives no relocation-specific
training, whereas each warm-start model receives 10,000 adaptation
updates. Thus reusing the global weights does not recover the local
accuracy retained by the translated modules under this recipe.

\paragraph{Full-scene and local spatial comparisons}
Figure~\ref{fig:phase3_full_scene} first compares the complete
\(1000~\mathrm{m}^2\), 50-obstacle fixed-density scene after relocation,
at \(\tau=1~\mathrm{s}\) and \(\theta=\pi/2\). The composed singleton
network retains the repeated obstacle-relative value structures;
placing the frozen local modules reduces the error around interacting
pairs without fitting a new scene-wide network. The global scratch
and neural weight-warm-start predictions exhibit more spatially
widespread signed errors in this displayed slice. The respective
full-scene slice RMSEs are \(0.0266\) for composition, \(0.0183\)
for reused patches, \(0.1361\) for global scratch, and \(0.1312\)
for global warm start. These values describe one fixed
heading--horizon slice over the displayed spatial grid; they must not
be conflated with the multi-heading, interaction-local aggregates in
Table~\ref{tab:phase3_local_scaling}.

\begin{figure*}[!t]
    \centering
    \includegraphics[width=\textwidth]{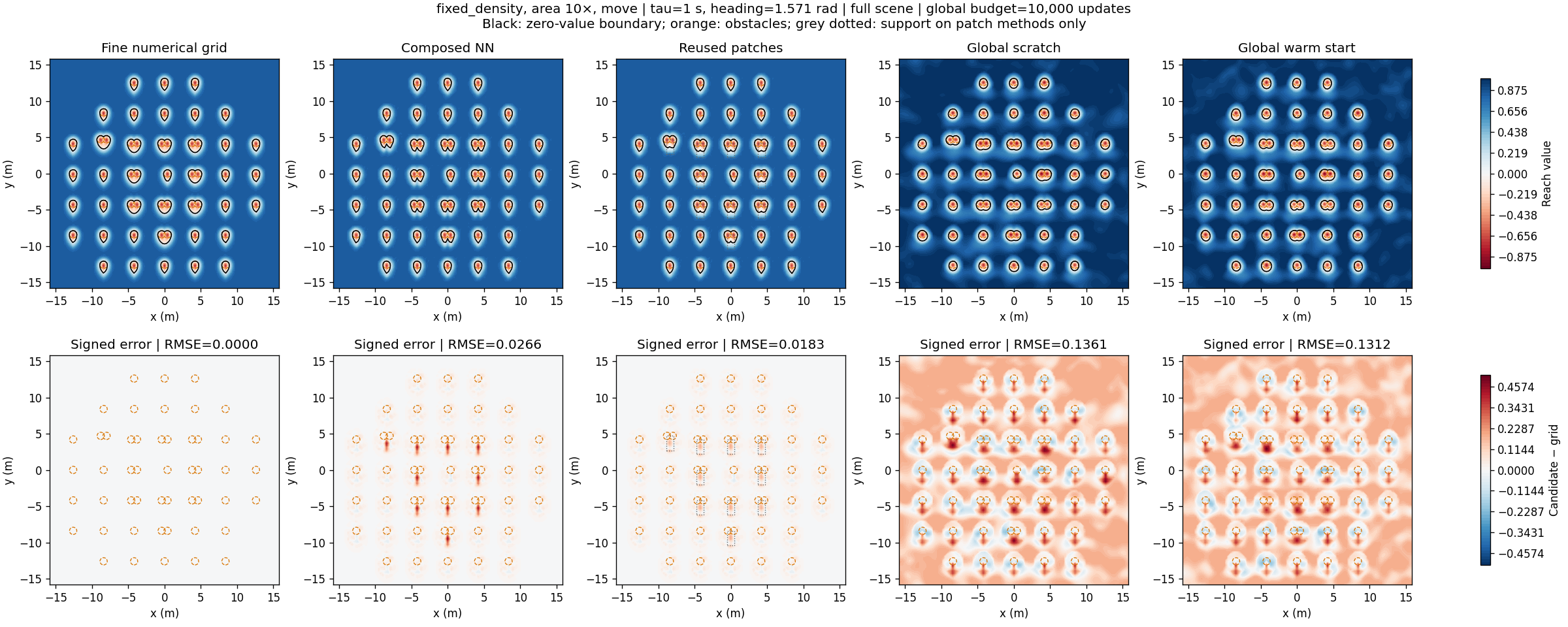}
    \caption{Full-workspace comparison after translation reuse in the
    \(1000~\mathrm{m}^2\), 50-obstacle fixed-density scene at
    \(\tau=1~\mathrm{s}\) and \(\theta=\pi/2\). Columns show the fine-grid
    numerical reference, composed singleton network, composition with
    reused local patches, global scratch network, and global neural
    warm-start network. The upper row shows predicted values and zero
    level sets (black); the lower row shows signed errors (candidate
    minus numerical reference) and the RMSE over this displayed slice.
    Orange outlines show obstacles, and dotted rectangles mark patch
    support where applicable. The global baselines receive 10,000
    relocation updates; the reusable components are frozen. The
    fixed-slice RMSEs should not be confused with the aggregate
    interaction errors in Table~\ref{tab:phase3_local_scaling}.}
    \label{fig:phase3_full_scene}
\end{figure*}

Figure~\ref{fig:phase3_representative} then magnifies a representative
interaction window from the same area/count setting at
\(\tau=1~\mathrm{s}\) and \(\theta=\pi/2\). The zoom highlights a
specific local contour discrepancy: singleton composition does not
fully reproduce the reference interaction shape, whereas the reused
local correction reduces the localized signed error. The global models
show larger discrepancies in this window under the reported recipe.
The zoom uses the recorded fixed seed-7 case; neither visualization
replaces the multi-seed aggregate comparisons in
Table~\ref{tab:phase3_local_scaling}.

\begin{figure*}[!t]
\centering
\includegraphics[width=\textwidth]{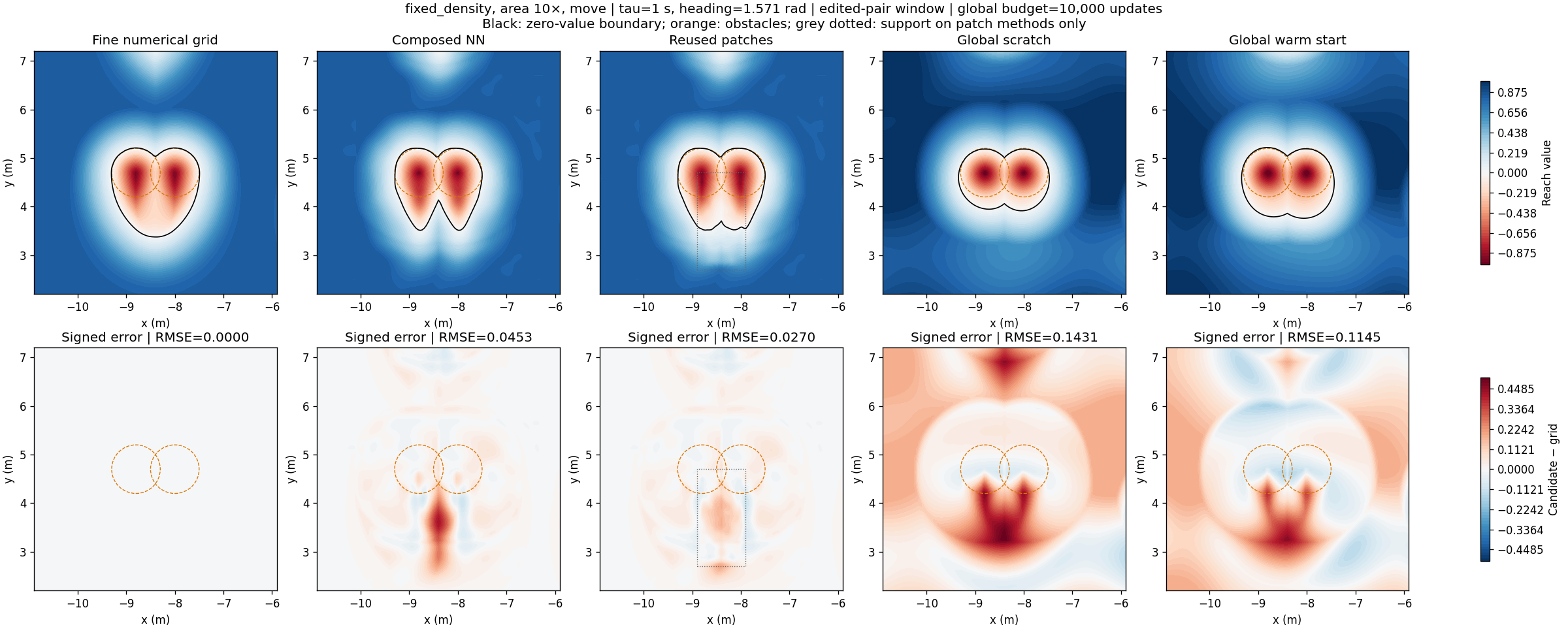}
\caption{Representative relocated interaction window for the
\(1000~\mathrm{m}^2\), 50-obstacle fixed-density scene, with fixed
seed 7, \(\tau=1~\mathrm{s}\), and \(\theta=\pi/2\).
Columns show the full numerical reference, singleton composition,
composition with the reused patch, global scratch, and global warm
start. Rows show values and signed errors; black curves are zero
contours, orange curves are obstacles, and the dotted rectangle is
patch support. This original fixed-slice visualization is not a
best-seed selection. Its slice RMSE differs from the aggregate
interaction population in Table~\ref{tab:phase3_local_scaling}.}
\label{fig:phase3_representative}
\end{figure*}

\paragraph{Warm starting versus training from scratch}
Figure~\ref{fig:adaptation_updates} answers a separate question:
whether previous global weights improve optimization after a scene
edit. Initial changed-region RMSE is 0.112 for warm start versus
0.149 for scratch in the \(3\times\) fixed-count scene, and
0.110 versus 0.180 in the \(10\times\) fixed-density scene.
The benefit persists through the first several hundred to few
thousand updates. At 10,000 updates, scratch has lower changed-region
RMSE in all three fixed-count cases and the \(6\times\) fixed-density
case; warm start is better in the other two. Accordingly, warm
starting is useful, but it is not equivalent to retaining the
represented functions under translation. The changed-region probes
are distinct from the obstacle-local and interaction populations, so
we do not draw a modular reference line from unmatched table values.

Taken together, the results support a specific empirical advantage:
for the tested translated geometry and training recipe, frozen
modular reuse yields substantially lower local error than global
neural weight warm starting, without scene-specific adaptation.
The library's additional stored capacity and repeated component
queries are accounted for separately below; this is not an
 equal-capacity or equal-compute comparison, nor a comparison against
 every HJ warm-start method.

\begin{figure*}[!t]
    \centering
    \includegraphics[width=\textwidth]{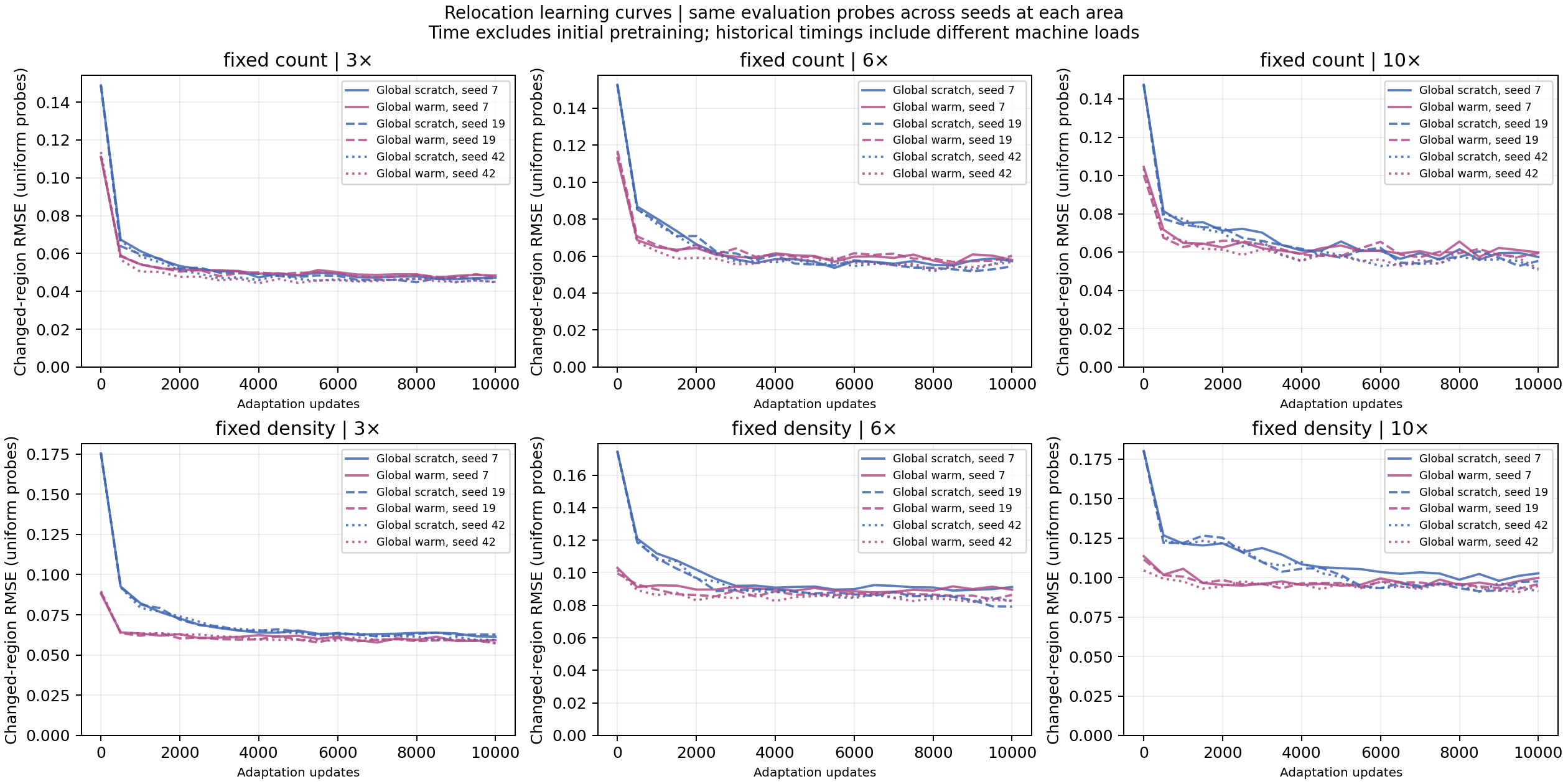}
    \caption{Global scratch and neural weight-warm-start adaptation on
    the changed-region probes. Rows show fixed count and fixed density;
    columns show \(300,600,1000~\mathrm{m}^2\). Individual curves use
    seeds \(7,19,42\); the horizontal axis excludes initial-scene
    training. Warm starting improves early optimization, but the final
    scratch--warm ordering varies. These probes differ from the local
    populations in Table~\ref{tab:phase3_local_scaling}; the modular
    comparison uses that table, rather than an unmatched reference line.}
    \label{fig:adaptation_updates}
\end{figure*}

\subsection{Offline, Adaptation, and Evaluation Cost}
\label{subsec:results_cost}

The library requires 10,000 singleton and 10,000 patch updates,
incurred once across the study. Translating the components adds zero
training updates but still requires assembly and inference. A global
initial-plus-relocation lifecycle uses 10,000 initial and 10,000
adaptation updates. These counts are not equal-compute budgets;
complete historical library timing is unavailable, so they do not
establish an end-to-end speedup.

The library stores 34,818 unique parameters, compared with 17,409
for one global network. Shared weights remove the need to store a
separate network per obstacle, but the number of template and patch
evaluations grows with their placements. Modular evaluation can thus
cost more than one global-NN query. The demonstrated computational
advantage is zero incremental training for the tested translations,
with an explicit storage and inference tradeoff.

\subsection{Physical Robot Feasibility}
\label{subsec:results_hardware}

One representative recording spans \(324.6~\mathrm{s}\), with
8,738 issued commands, three completed laps, and 15 reached waypoints.
The run summary reports a mean achieved command rate of
\(26.96~\mathrm{Hz}\) for a requested \(30~\mathrm{Hz}\).
The saved run note reports a
 dashboard timeout at termination, followed by an acknowledged stop;
 the recording is not described as planned completion of the route.
 The first \(45~\mathrm{s}\) shown in Fig.~\ref{fig:picarx_hardware}
 precede that termination. Over that excerpt, cone-center variation
 is below \(1.4~\mathrm{mm}\) per coordinate, so the arrangement is
 treated as stationary. Across approximately 20
proof-of-concept runs over multiple cone layouts, no collision was
observed. These trials demonstrate online reuse of a frozen singleton
for several physical obstacles without scene-specific value training.
They are feasibility demonstrations at one speed setting, not a
statistically controlled benchmark or a formal closed-loop safety
certificate.


\subsection{Limitations}
\label{subsec:results_limitations}

The numerical study uses one singleton and one patch trained for an
overlapping pair of circular obstacles. Translation reuse preserves
that pair's separation and orientation; new shapes and interaction
geometries are not tested. The three global-training seeds do not
replicate library construction or layouts, and standalone singleton
validation covers a fixed heading. Grid references remain numerical
approximations; a quantitative grid-convergence study is not reported.
The missed-avoid metric also omits the opposite false-avoid error.
Because the raw residual is subtractive, a nonincreasing evaluation
rule cannot increase the missed-avoid rate. A reduction in that rate
alone therefore cannot establish improved set accuracy; the accompanying value-error
reductions are therefore essential.

Global comparisons use a fixed architecture and update budget rather
than matching total capacity or training to convergence. They do not
exclude an optimized global model attaining comparable accuracy.
No formal neural-error or hardware safety certificate is established.
Within this scope, the evidence supports separating reusable singleton
structure from local interaction correction and relocating the tested
components without additional training.

\section{Conclusion}
\label{sec:conclusion}

This work separates reusable singleton reachability information from
configuration-dependent control conflicts. Translation gives exact
function-level reuse, whereas minimum composition is optimistic
because it permits different avoidance controls for different
obstacles. Local continuation problems recover part of the missing
common-control information while preserving an upper bound and
leaving the exterior unchanged. Successive exact refinements enlarge
an inner certificate of unavoidable collision, and reachable-tube
containment supplies a sufficient condition for local exactness.
The conditional uniform-error bound additionally shows that minimum
composition does not amplify singleton approximation error with
obstacle count.

The neural realization preserves this modular organization by
freezing the singleton and learning localized, nonnegative
corrections. In the tested Dubins scenes, the patch reduces
inside-region RMSE by 46.3\% and near-zero RMSE by 58.5\%.
The same frozen library retains similar local accuracy after
translation into \(300\)--\(1000~\mathrm{m}^2\) workspaces.
Across all six relocated cases, interaction RMSE is approximately
63--73\% lower than the mean global neural warm-start baseline,
with zero relocation training rather than 10,000 adaptation updates.
Warm starting still improves early global optimization relative to
scratch training. The comparison therefore distinguishes retaining
reusable functions from reusing weights as an optimization
initialization, under the reported geometry and training recipe.
The PiCar-X demonstration provides complementary evidence that
singleton reuse can be integrated with PID tracking and a
look-ahead avoidance filter in a physical multi-obstacle setting.

Exact operator bounds do not automatically certify learned networks
or the physical controller. The numerical results concern one
frozen library on repeated pair geometry, and inference cost grows
with component queries. Broader interaction geometries, automatic
module selection, stronger matched-resource baselines, and certified
approximation remain directions for further work. The central design
principle is to preserve structure already represented and add local
capacity for interactions that translation alone cannot supply.

\bibliographystyle{IEEEtran}
\bibliography{references}

\end{document}